\documentclass[a4paper]{iopart}
\expandafter\let\csname equation*\endcsname\relax
\expandafter\let\csname endequation*\endcsname\relax
\usepackage[fleqn]{mathtools}
\usepackage{amssymb}
\usepackage{amsthm}
\usepackage{stix}
\usepackage{inconsolata}
\usepackage{bm}
\usepackage{mleftright}
\mleftright
\usepackage{cite}
\usepackage[pdftex,colorlinks,linkcolor=blue,citecolor=blue,urlcolor=blue]{hyperref}
\newtheorem{theorem}{Theorem}[section]
\newtheorem{corollary}{Corollary}[section]
\def\keywords#1{\vspace*{5pt}\address{Keywords: #1}}
\begin{document}
\suppressfloats
\title{Nonautonomous ultradiscrete hungry Toda lattice and a generalized box--ball system}
\author{Kazuki Maeda}
\address{Department of Mathematical Sciences, School of Science and Technology, Kwansei Gakuin University, 2-1 Gakuen, Sanda 669-1337, Japan}
\ead{kmaeda@kmaeda.net}
\begin{abstract}
  A nonautonomous version of the ultradiscrete hungry Toda lattice
  with a finite lattice boundary condition is derived by applying
  reduction and ultradiscretization to a nonautonomous
  two-dimensional discrete Toda lattice.
  It is shown that the derived ultradiscrete system has a direct connection to
  the box--ball system with many kinds of balls and finite carrier capacity.
  Particular solutions to the ultradiscrete system are constructed
  by using the theory of some sort of discrete biorthogonal polynomials.
\end{abstract}

\keywords{integrable systems, biorthogonal polynomials, discrete two-dimensional Toda lattice}
\pacs{02.30.Ik, 05.45.Yv}

\section{Introduction}

The box--ball system (BBS) is a soliton cellular automaton
composed of an infinite array of boxes and a finite number of balls~\cite{takahashi1990sca},
known as one of the most important ultradiscrete integrable systems.
The time evolution equation of the original BBS
\begin{equation}\label{eq:orig-BBS}
  U^{(t+1)}_n=\min\left(1-U^{(t)}_n, \sum_{j=-\infty}^{n-1} (U^{(t)}_j-U^{(t+1)}_j)\right),
\end{equation}
where $U^{(t)}_n \in \{0, 1\}$ denotes the number of balls in the $n$th box at time $t$,
is derived from the discrete
KdV lattice through ultradiscretization~\cite{tokihiro1996fse,tsujimoto1998uke}.
\begin{figure}
  \centering
  \includegraphics{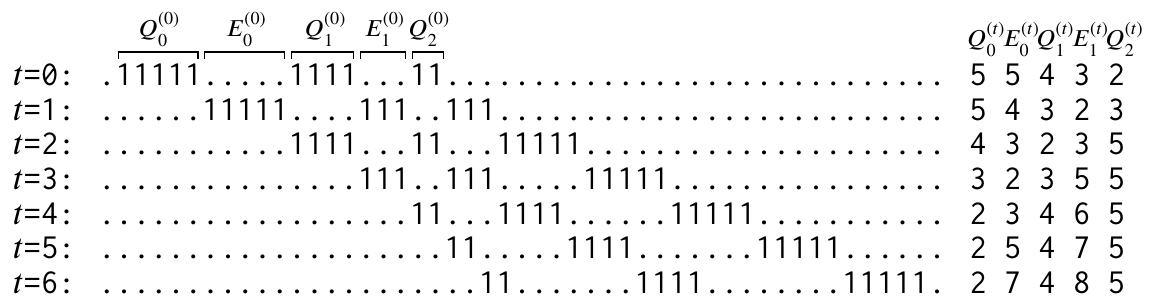}
  \caption{An example of the time evolution of (left) the original BBS~\eqref{eq:orig-BBS}
    and (right) the ultradiscrete Toda lattice with the finite lattice boundary
    condition~\eqref{eq:u-Toda}.}
  \label{fig:bbs}
\end{figure}
The left side of \fref{fig:bbs} shows an example of the time evolution of
the original BBS~\eqref{eq:orig-BBS}, in which `\texttt{1}' and `\texttt{.}' denotes
a ball and an empty box, respectively.
We can observe that three blocks of balls move from left to right
and interact with each other like solitons.

There is another time evolution equation for the original BBS:
\begin{subequations}\label{eq:u-Toda}
  \begin{gather}
    Q^{(t+1)}_n=\min\left(E^{(t)}_n, \sum_{j=0}^{n}Q^{(t)}_j-\sum_{j=0}^{n-1} Q^{(t+1)}_j\right),\\
    E^{(t+1)}_n=E^{(t)}_n-Q^{(t+1)}_n+Q^{(t)}_{n+1}
  \end{gather}
  for $n=0, 1, \dots, N-1$ with the finite lattice boundary condition
  \begin{equation}
    E^{(t)}_{-1}=E^{(t)}_{N-1}=+\infty
  \end{equation}
  for all $t \in \mathbb Z$, where
\end{subequations}
\begin{itemize}
\item $Q^{(t)}_n$: the number of balls in the $n$th block at time $t$;
\item $E^{(t)}_n$: the number of empty boxes between the $n$th and $(n+1)$st blocks of balls at time $t$;
\item $N$: the number of the blocks of balls.
\end{itemize}
It is known that equations~\eqref{eq:u-Toda} are derived from
the discrete Toda lattice
\begin{gather*}
  q^{(t+1)}_n+e^{(t+1)}_{n-1}=q^{(t)}_n+e^{(t)}_n,\\
  q^{(t+1)}_n e^{(t+1)}_n=q^{(t)}_{n+1} e^{(t)}_n
\end{gather*}
with the finite lattice condition
\begin{equation*}
  e^{(t)}_{-1}=e^{(t)}_{N-1}=0
\end{equation*}
through ultradiscretization~\cite{nagai1999sca}.
The right side of \fref{fig:bbs} shows an example of the time evolution of
the ultradiscrete Toda lattice~\eqref{eq:u-Toda},
in which the initial values are chosen to correspond to
the initial state of the original BBS on the left side.

We can introduce some extended rules to the original BBS.
The time evolution equations of the extended BBSs are derived from the nonautonomous
discrete KP lattice through reduction and ultradiscretization,
or from the geometric crystal for $\widehat{\mathfrak{sl}}_{M+1}$
through crystallization~\cite{hatayama2001taa,inoue2012isb}.
We have already known the following correspondences
between the BBS with an extended rule and an ultradiscrete Toda type system:
\begin{enumerate}
\item The BBS with many kinds of balls and the ultradiscrete hungry Toda lattice~\cite{tokihiro1999pos};
\item The BBS with a carrier of balls whose capacity is finite
  and the nonautonomous ultradiscrete Toda lattice~\cite{maeda2010bbs};
\item The BBS with boxes whose capacity is greater than one and
  a generalized ultradiscrete Toda lattice~\cite{maeda2012ftr}.
\end{enumerate}
In addition, there are studies on the relation between
the BBS and the ultradiscrete Toda lattice
for the case of a periodic boundary condition~\cite{idzumi2009siv},
and for the case in which the number of balls in each box can take any real value~\cite{gilson2015dau}.

In this paper, we wish to derive and study an ultradiscrete Toda type system
corresponding to the BBS with both the rules (i) and (ii).
To this end, in \sref{sec:biorth-polyn}, we first consider the theory of
biorthogonal polynomials and derive a nonautonomous version of
the discrete two-dimensional Toda lattice (nd-2D-Toda lattice).
It is known that the theory of (bi)orthogonal functions
is a very useful tool for deriving and analyzing many Toda type systems
and their solutions~\cite{spiridonov1995ddt,kharchev1997frt,spiridonov1997dtv,spiridonov2000stc,spiridonov2007idt,miki2012dst,maeda2013dcr,aptekarev2016mtl}.
In addition to these known results,
we will show, based on the previous studies~\cite{iserles1988tbp,adler1997sop,tsujimoto2010dso,tsujimoto2000msde},
that the nd-2D-Toda lattice is derived as compatibility conditions for spectral transformations
of biorthogonal polynomials.
Since the discrete two-dimensional Toda lattice hierarchy
yields many discrete Toda type systems by imposing reduction conditions,
we will be able to derive many nonautonomous discrete Toda type systems
from the nd-2D-Toda lattice hierarchy in the same manner.
In \sref{sec:m-1-reduction}, we will impose $(M, 1)$-reduction condition
for biorthogonal polynomials and the nd-2D-Toda lattice,
and derive a nonautonomous version of the discrete hungry Toda lattice (ndh-Toda lattice).
Further, we will impose a finite lattice boundary condition to the ndh-Toda lattice
and give a particular solution coming from a determinant structure of the biorthogonal
polynomials. We will also give a condition for the positivity of the solution.
In \sref{sec:ultradiscretization}, we will ultradiscretize the ndh-Toda lattice
and its solution, and prove that the derived ultradiscrete system
is another time evolution equation of the generalized BBS.
\Sref{sec:concluding-remarks} is devoted to concluding remarks.

\section{Biorthogonal polynomials and semi-infinite lattice equations}\label{sec:biorth-polyn}

In this section, we consider the theory of biorthogonal polynomials,
and derive the nd-2D-Toda lattice
with a semi-infinite lattice boundary condition as compatibility conditions
for spectral transformations of the biorthogonal polynomials.

\subsection{Definitions and determinant representations}
Let $\mathcal B\colon \mathbb C[z]\times \mathbb C[z] \to \mathbb C$ be
a bilinear form. Let us consider polynomial sequences $\{\phi_n(z)\}_{n=0}^\infty$
and $\{\psi_n(z)\}_{n=0}^\infty$ satisfying the following properties:
\begin{enumerate}\renewcommand{\labelenumi}{(\roman{enumi})}
\item $\deg\phi_n(z)=\deg\psi_n(z)=n$;
\item The polynomials $\phi_n(z)$ and $\psi_n(z)$ are monic; i.e.
  the leading coefficients of $\phi_n(z)$ and $\psi_n(z)$ are one;
\item The \emph{biorthogonal relation} with respect to $\mathcal B$
  \begin{equation}\label{eq:biorthogonal-relation}
    \mathcal B[\phi_m(z), \psi_n(z)]=h_n \delta_{m, n},\quad
    h_n \ne 0,\quad
    m, n=0, 1, 2, \dots,
  \end{equation}
  holds, where $\delta_{m, n}$ is the Kronecker delta.
\end{enumerate}
We call the polynomial sequences $\{\phi_n(z)\}_{n=0}^\infty$
and $\{\psi_n(z)\}_{n=0}^\infty$
the pair of monic \emph{biorthogonal polynomial sequences} with
respect to $\mathcal B$.

Note that, since both $\{\phi_n(z)\}_{n=0}^\infty$
and $\{\psi_n(z)\}_{n=0}^\infty$ span $\mathbb C[z]$,
the biorthogonal relation~\eqref{eq:biorthogonal-relation} is equivalent to
\begin{multline}\label{eq:biorthogonal-relation-2}
  \mathcal B[\phi_n(z), z^m]=h_n \delta_{m, n},\quad
  \mathcal B[z^m, \psi_n(z)]=h_n \delta_{m, n},\\
  n=0, 1, 2, \dots, \quad
  m=0, 1, \dots, n.
\end{multline}

\begin{theorem}\label{th:biorth-det-repr}
  The pair of monic biorthogonal polynomial sequences
  $\{\phi_n(z)\}_{n=0}^\infty$ and $\{\psi_n(z)\}_{n=0}^\infty$ has
  the determinant representations
  \begin{subequations}\label{eq:det-expression}
    \begin{flalign}
      \phi_0(z)&=1,&
      \phi_n(z)&=\frac{1}{\tau_n}
      \begin{vmatrix}
        \mu_{0, 0} & \mu_{0, 1} & \dots & \mu_{0, n-1} & 1\\
        \mu_{1, 0} & \mu_{1, 1} & \dots & \mu_{1, n-1} & z\\
        \vdots & \vdots & & \vdots & \vdots\\
        \mu_{n-1, 0} & \mu_{n-1, 1} & \dots & \mu_{n-1, n-1} & z^{n-1}\\
        \mu_{n, 0} & \mu_{n, 1} & \dots & \mu_{n, n-1} & z^{n}
      \end{vmatrix},&& n=1, 2, 3, \dots,\\
      \psi_0(z)&=1,&
      \psi_n(z)&=\frac{1}{\tau_n}
      \begin{vmatrix}
        \mu_{0, 0} & \mu_{0, 1} & \dots & \mu_{0, n-1} & \mu_{0, n}\\
        \mu_{1, 0} & \mu_{1, 1} & \dots & \mu_{1, n-1} & \mu_{1, n}\\
        \vdots & \vdots & & \vdots & \vdots\\
        \mu_{n-1, 0} & \mu_{n-1, 1} & \dots & \mu_{n-1, n-1} & \mu_{n-1, n}\\
        1 & z & \dots & z^{n-1} & z^n
      \end{vmatrix},&\ & n=1, 2, 3, \dots,
    \end{flalign}
  \end{subequations}
  where $\mu_{i, j}$ is the moment of $\mathcal B$ defined by
  \begin{equation*}
    \mu_{i, j}\coloneq\mathcal B[z^i, z^j],\quad i, j=0, 1, 2, \dots,
  \end{equation*}
  and $\tau_n$ is the determinant whose entries are the moments:
  \begin{equation*}
    \tau_0\coloneq1,\quad
    \tau_n\coloneq|\mu_{i, j}|_{i, j=0}^{n-1},\quad n=1, 2, 3, \dots.
  \end{equation*}
  Here, we assume that $\tau_n \ne 0$ for all $n=1, 2, 3, \dots$.
  The constant $h_n$ in the biorthogonal relation~\eqref{eq:biorthogonal-relation-2}
  is given by
  \begin{equation*}
    h_n=\frac{\tau_{n+1}}{\tau_n}.
  \end{equation*}
\end{theorem}
\begin{proof}
  Let $c_{n, i}$ be the coefficients of the polynomial $\phi_n(z)$:
  \begin{equation*}
    \phi_n(z)=z^n+\sum_{i=0}^{n-1} c_{n, i} z^i.
  \end{equation*}
  Then, the biorthogonal relation~\eqref{eq:biorthogonal-relation-2} gives
  the linear equation
  \begin{equation*}
    \begin{pmatrix}
      \mu_{0, 0} & \mu_{1, 0} & \dots & \mu_{n-1, 0} & \mu_{n, 0}\\
      \mu_{0, 1} & \mu_{1, 1} & \dots & \mu_{n-1, 1} & \mu_{n, 1}\\
      \vdots & \vdots & & \vdots & \vdots\\
      \mu_{0, n-1} & \mu_{1, n-1} & \dots & \mu_{n-1, n-1} & \mu_{n, n-1}\\
      \mu_{0, n} & \mu_{1, n} & \dots & \mu_{n-1, n} & \mu_{n, n}
    \end{pmatrix}
    \begin{pmatrix}
      c_{n, 0}\\
      c_{n, 1}\\
      \vdots\\
      c_{n, n-1}\\
      1
    \end{pmatrix}=
    \begin{pmatrix}
      0\\
      0\\
      \vdots\\
      0\\
      h_n
    \end{pmatrix}.
  \end{equation*}
  Applying Cramer's rule yields the relation $1=\tau_n h_n/\tau_{n+1}$,
  which implies $h_n=\tau_{n+1}/\tau_n$, and
  the determinant representation of $\phi_n(z)$.
  The determinant representation of $\psi_n(z)$ is also given in the same manner.
\end{proof}

The proof of Theorem~\ref{th:biorth-det-repr} also claims that,
for a bilinear form $\mathcal B$ on $\mathbb C[z]$,
there is a \emph{unique} pair of monic biorthogonal polynomial sequences
$\{\phi_n(z)\}_{n=0}^\infty$ and $\{\psi_n(z)\}_{n=0}^\infty$ if exists.

\subsection{Spectral transformations}
From the pair of the monic biorthogonal polynomial sequences $\{\phi_n(z)\}_{n=0}^\infty$
and $\{\psi_n(z)\}_{n=0}^\infty$ with respect to $\mathcal B$,
we can construct new monic biorthogonal polynomials
\begin{subequations}\label{eq:Christoffel}
  \begin{alignat}{2}
    &\phi^*_n(z)\coloneq\frac{\phi_{n+1}(z)+q^*_n\phi_n(z)}{z-s^*},&\quad&
    q^*_n\coloneq-\frac{\phi_{n+1}(s^*)}{\phi_n(s^*)},\\
    &\psi^\dagger_n(z)\coloneq\frac{\psi_{n+1}(z)+q^\dagger_n\psi_n(z)}{z-s^\dagger},&&
    q^\dagger_n\coloneq-\frac{\psi_{n+1}(s^\dagger)}{\psi_n(s^\dagger)},
  \end{alignat}
\end{subequations}
for $n=0, 1, 2, \dots$,
where $s^*$ and $s^\dagger$ are any complex parameters satisfying
$\phi_n(s^*)\ne 0$ and $\psi_n(s^\dagger)\ne 0$ for all $n$.
Since $\phi_{n+1}(s^*)+q^*_n\phi_n(s^*)=0$ and
$\psi_{n+1}(s^\dagger)+q^\dagger_n \psi_n(s^\dagger)=0$,
both $\phi^*_n(z)$ and $\psi^\dagger_n(z)$ are monic $n$th degree polynomials.
The new polynomials are biorthogonal with respect to the bilinear forms
$\mathcal B^*$ and $\mathcal B^\dagger$ defined by
\begin{multline}\label{eq:bilinear-rel}
  \mathcal B^*[z^i, z^j]\coloneq\mathcal B[(z-s^*)z^i, z^j],\quad
  \mathcal B^\dagger[z^i, z^j]\coloneq\mathcal B[z^i, (z-s^\dagger)z^j],\\
  i, j=0, 1, 2, \dots.
\end{multline}
We can readily verify that the biorthogonal relations
\begin{gather*}
  \mathcal B^*[\phi^*_n(z), z^m]
  =\mathcal B[\phi_{n+1}(z)+q^*_n\phi_n(z), z^m]
  =\frac{q^*_n\tau_{n+1}}{\tau_n} \delta_{m, n},\\
  \mathcal B^\dagger[z^m, \psi^\dagger_n(z)]
  =\mathcal B[z^m, \psi_{n+1}(z)+q^\dagger_n\psi_n(z)]
  =\frac{q^\dagger_n\tau_{n+1}}{\tau_n} \delta_{m, n}
\end{gather*}
indeed hold for $n=0, 1, 2, \dots$ and $m=0, 1, \dots, n$.
The relations~\eqref{eq:Christoffel} are analogues of the Christoffel transformation
for monic orthogonal polynomials.

Next, let us derive a relation between the polynomial sequences $\{\phi_n(z)\}_{n=0}^\infty$ and
$\{\phi^\dagger_n(z)\}_{n=0}^\infty$ satisfying the biorthogonal relation
\begin{equation*}
  \mathcal B^\dagger[\phi^\dagger_n(z), z^m]=\frac{q^\dagger_n\tau_{n+1}}{\tau_n}\delta_{m, n},\quad
  n=0, 1, 2, \dots,\quad m=0, 1, \dots, n.
\end{equation*}
For each $n=0, 1, 2, \dots$, write $\phi_{n+1}(z)$ as a linear combination of
$\phi^\dagger_0(z), \phi^\dagger_1(z), \dots, \phi^\dagger_{n+1}(z)$:
\begin{equation*}
  \phi_{n+1}(z)=\phi^\dagger_{n+1}(z)+\sum_{i=0}^{n} c_{n+1, i} \phi^\dagger_i(z).
\end{equation*}
Then we have, if $n \ge 1$,
\begin{gather*}
  \mathcal B[\phi_{n+1}(z), z-s^\dagger]
  =\mathcal B^\dagger\left[\phi^\dagger_{n+1}(z)+\sum_{i=0}^{n} c_{n+1, i} \phi^\dagger_i(z), 1\right]
  =c_{n+1, 0}\frac{q^\dagger_0 \tau_1}{\tau_0}=0.
\end{gather*}
Since $q^\dagger_0 \tau_1/\tau_0 \ne 0$, that implies $c_{n+1, 0}=0$.
In the same manner,
the equation $\mathcal B[\phi_{n+1}(z), (z-s^\dagger)z^m]=0$ implies $c_{n+1, m}=0$,
$m=0, 1, \dots, n-1$, by induction on $m$. Finally,
\begin{equation*}
  \fl
  \mathcal B[\phi_{n+1}(z), (z-s^\dagger)z^{n}]
  =\mathcal B^\dagger[\phi^\dagger_{n+1}(z)+c_{n+1, n}\phi^\dagger_{n}(z), z^{n}]
  =c_{n+1, n}\frac{q^\dagger_{n}\tau_{n+1}}{\tau_n}=\frac{\tau_{n+2}}{\tau_{n+1}}.
\end{equation*}
Thus $c_{n+1, n}=\tau_n \tau_{n+2}/q^\dagger_n(\tau_{n+1})^2\ne 0$.
We can show by similar discussions that the relation
\begin{equation*}
  \psi_{n+1}(z)=\psi^*_{n+1}(z)+\frac{\tau_n \tau_{n+2}}{q^*_{n}(\tau_{n+1})^2}\psi^*_{n}(z)
\end{equation*}
also holds.

We summarize the results above.

\begin{theorem}\label{th:st}
  Let $\mathcal B$, $\mathcal B^*$ and $\mathcal B^\dagger$ be bilinear
  forms on $\mathbb C[z]$ connected by the relations~\eqref{eq:bilinear-rel}.
  Suppose that $\{\phi_n(z)\}_{n=0}^\infty$ and $\{\psi_n(z)\}_{n=0}^\infty$,
  $\{\phi^*_n(z)\}_{n=0}^\infty$ and $\{\psi^*_n(z)\}_{n=0}^\infty$,
  and $\{\phi^\dagger_n(z)\}_{n=0}^\infty$ and $\{\psi^\dagger_n(z)\}_{n=0}^\infty$
  are the pairs of the monic biorthogonal polynomial sequences with respect to
  $\mathcal B$, $\mathcal B^*$, and $\mathcal B^\dagger$, respectively.
  Then, these polynomial sequences satisfy the following relations:
  \begin{subequations}\label{eq:st}
    \begin{alignat}{2}
      (z-s^*)\phi^*_n(z)&=\phi_{n+1}(z)+q^*_n \phi_n(z),&\quad&
      q^*_n=-\frac{\phi_{n+1}(s^*)}{\phi_n(s^*)},\\
      \phi_{n+1}(z)&=\phi^\dagger_{n+1}(z)+e^\dagger_n \phi^\dagger_n(z),&&
      e^\dagger_n=\frac{\tau_n \tau_{n+2}}{q^\dagger_{n}(\tau_{n+1})^2},\label{eq:Geronimus-1}\\
      (z-s^\dagger)\psi^\dagger_n(z)&=\psi_{n+1}(z)+q^\dagger_n \psi_n(z),&\quad&
      q^\dagger_n=-\frac{\psi_{n+1}(s^\dagger)}{\psi_n(s^\dagger)},\\
      \psi_{n+1}(z)&=\psi^*_{n+1}(z)+e^*_n \psi^*_n(z),&&
      e^*_n=\frac{\tau_n \tau_{n+2}}{q^*_{n}(\tau_{n+1})^2},\label{eq:Geronimus-2}
    \end{alignat}
  \end{subequations}
  for $n=0, 1, 2, \dots$.
\end{theorem}

We should remark that the relations~\eqref{eq:Geronimus-1} and \eqref{eq:Geronimus-2}
are analogues of the Geronimus transformation for monic orthogonal polynomials.

Let us define the moment of $\mathcal B^*$ and $\mathcal B^\dagger$
\begin{equation*}
  \mu^*_{i, j}\coloneq\mathcal B^*[z^i, z^j],\quad
  \mu^\dagger_{i, j}\coloneq \mathcal B^\dagger[z^i, z^j],\quad
  i, j=0, 1, 2, \dots,
\end{equation*}
and the determinant
\begin{equation*}
  \tau^*_0\coloneq1,\quad \tau^\dagger_0\coloneq1,\quad
  \tau^*_n\coloneq|\mu^*_{i, j}|_{i, j=0}^{n-1},\quad
  \tau^\dagger_n\coloneq|\mu^\dagger_{i, j}|_{i, j=0}^{n-1},\quad
  n=1, 2, 3, \dots.
\end{equation*}
Then, from \eqref{eq:bilinear-rel}, we have
\begin{equation}\label{eq:dispersion}
  \mu^*_{i, j}=\mu_{i+1, j}-s^* \mu_{i, j},\quad
  \mu^\dagger_{i, j}=\mu_{i, j+1}-s^\dagger \mu_{i, j}.
\end{equation}
Applying the elementary column-additions (or row-additions) to the determinant representations
of the monic biorthogonal polynomials~\eqref{eq:det-expression}
and using the relations~\eqref{eq:dispersion},
we have
\begin{equation*}
  \phi_n(s^*)=(-1)^n \frac{\tau^*_n}{\tau_n},\quad
  \psi_n(s^\dagger)=(-1)^n \frac{\tau^\dagger_n}{\tau_n}.
\end{equation*}
Hence, the variables appear in \eqref{eq:st} are rewritten as
\begin{equation*}
  q^*_n=\frac{\tau_n \tau^*_{n+1}}{\tau_{n+1}\tau^*_n},\quad
  e^\dagger_n=\frac{\tau_{n+2}\tau^\dagger_n}{\tau_{n+1}\tau^\dagger_{n+1}},\quad
  q^\dagger_n=\frac{\tau_n \tau^\dagger_{n+1}}{\tau_{n+1}\tau^\dagger_n},\quad
  e^*_n=\frac{\tau_{n+2}\tau^*_n}{\tau_{n+1}\tau^*_{n+1}}.
\end{equation*}

\subsection{Nonautonomous discrete semi-infinite two-dimensional Toda lattice}
Let us introduce discrete time variables $k_1, k_2, t_1, t_2$
into bilinear forms as follows:
\begin{gather*}
  \mathcal B^{(k_1+1, k_2, t_1, t_2)}[z^i, z^j]\coloneq\mathcal B^{(k_1, k_2, t_1, t_2)}[z^{i+1}, z^{j}],\\
  \mathcal B^{(k_1, k_2+1, t_1, t_2)}[z^i, z^j]\coloneq\mathcal B^{(k_1, k_2, t_1, t_2)}[z^{i}, z^{j+1}],\\
  \mathcal B^{(k_1, k_2, t_1+1, t_2)}[z^i, z^j]\coloneq\mathcal B^{(k_1, k_2, t_1, t_2)}[(z-s_1^{(t_1)})z^{i}, z^{j}],\\
  \mathcal B^{(k_1, k_2, t_1, t_2+1)}[z^i, z^j]\coloneq\mathcal B^{(k_1, k_2, t_1, t_2)}[z^{i}, (z-s_2^{(t_2)})z^{j}]
\end{gather*}
for all $i, j=0, 1, 2, \dots$,
where $s_1^{(t_1)}$ and $s_2^{(t_2)}$ are parameters chosen at
each $t_1$ and $t_2$, respectively.
Then, the moment
\begin{equation*}
  \mu^{(t_1, t_2)}_{i, j}\coloneq\mathcal B^{(0, 0, t_1, t_2)}[z^i, z^j]
\end{equation*}
has the relations
\begin{gather*}
  \mathcal B^{(k_1, k_2, t_1, t_2)}[z^i, z^j]=\mu^{(t_1, t_2)}_{k_1+i, k_2+j},\\
  \mu^{(t_1+1, t_2)}_{i, j}=\mu^{(t_1, t_2)}_{i+1, j}-s_1^{(t_1)}\mu^{(t_1, t_2)}_{i, j},\quad
  \mu^{(t_1, t_2+1)}_{i, j}=\mu^{(t_1, t_2)}_{i, j+1}-s_2^{(t_2)}\mu^{(t_1, t_2)}_{i, j}.
\end{gather*}

Let $\{\phi^{(k_1, k_2, t_1, t_2)}_n(z)\}_{n=0}^\infty$ be
one of the pair of the monic biorthogonal polynomial sequences with respect to $\mathcal B^{(k_1, k_2, t_1, t_2)}$:
\begin{multline*}
  \mathcal B^{(k_1, k_2, t_1, t_2)}[\phi^{(k_1, k_2, t_1, t_2)}_n(z), z^m]=\frac{\tau^{(k_1, k_2, t_1, t_2)}_{n+1}}{\tau^{(k_1, k_2, t_1, t_2)}_{n}}\delta_{m, n},\\
  n=0, 1, 2, \dots, \quad m=0, 1, \dots, n,
\end{multline*}
where
\begin{equation*}
  \tau^{(k_1, k_2, t_1, t_2)}_0\coloneq1,\quad
  \tau^{(k_1, k_2, t_1, t_2)}_n\coloneq|\mu^{(t_1, t_2)}_{k_1+i, k_2+j}|_{i, j=0}^{n-1},\quad
  n=1, 2, 3, \dots.
\end{equation*}
The polynomials $\{\phi^{(k_1, k_2, t_1, t_2)}_n(z)\}_{n=0}^\infty$ also have the relations
\begin{equation*}
  \phi^{(k_1, k_2, t_1, t_2)}_n(0)=\frac{\tau^{(k_1+1, k_2, t_1, t_2)}_n}{\tau^{(k_1, k_2, t_1, t_2)}_n},\quad
  \phi^{(k_1, k_2, t_1, t_2)}_n(s_1^{(t_1)})=\frac{\tau^{(k_1, k_2, t_1+1, t_2)}_n}{\tau^{(k_1, k_2, t_1, t_2)}_n}.
\end{equation*}
From Theorem~\ref{th:st}, the polynomials
$\{\phi^{(k_1, k_2, t_1, t_2)}_n(z)\}_{n=0}^\infty$ satisfy
\begin{subequations}\label{eq:st-b}
  \begin{flalign}
    z\phi^{(k_1+1, k_2, t_1, t_2)}_n(z)&=\phi^{(k_1, k_2, t_1, t_2)}_{n+1}(z)+q^{(k_1, k_2, t_1, t_2)}_n \phi^{(k_1, k_2, t_1, t_2)}_n(z),\label{eq:st-b-1}&\\
    \phi^{(k_1, k_2-1, t_1, t_2)}_{n+1}(z)&=\phi^{(k_1, k_2, t_1, t_2)}_{n+1}(z)+e^{(k_1, k_2-1, t_1, t_2)}_n \phi^{(k_1, k_2, t_1, t_2)}_n(z),\label{eq:st-b-2}&\\
    (z-s_1^{(t_1)})\phi^{(k_1, k_2, t_1+1, t_2)}_n(z)&=\phi^{(k_1, k_2, t_1, t_2)}_{n+1}(z)+\tilde q^{(k_1, k_2, t_1, t_2)}_n \phi^{(k_1, k_2, t_1, t_2)}_n(z),\label{eq:st-b-3}&\\
    \phi^{(k_1, k_2, t_1, t_2-1)}_{n+1}(z)&=\phi^{(k_1, k_2, t_1, t_2)}_{n+1}(z)+\tilde e^{(k_1, k_2, t_1, t_2-1)}_n \phi^{(k_1, k_2, t_1, t_2)}_n(z)&\label{eq:st-b-4}
  \end{flalign}
\end{subequations}
for $n=0, 1, 2, \dots$, where
\begin{xxalignat}{4}
  q^{(k_1, k_2, t_1, t_2)}_n=\frac{\tau^{(k_1, k_2, t_1, t_2)}_n\tau^{(k_1+1, k_2, t_1, t_2)}_{n+1}}{\tau^{(k_1, k_2, t_1, t_2)}_{n+1}\tau^{(k_1+1, k_2, t_1, t_2)}_n},&\quad&
  e^{(k_1, k_2, t_1, t_2)}_n=\frac{\tau^{(k_1, k_2, t_1, t_2)}_{n+2}\tau^{(k_1, k_2+1, t_1, t_2)}_n}{\tau^{(k_1, k_2, t_1, t_2)}_{n+1}\tau^{(k_1, k_2+1, t_1, t_2)}_{n+1}},\\
  \tilde q^{(k_1, k_2, t_1, t_2)}_n=\frac{\tau^{(k_1, k_2, t_1, t_2)}_n\tau^{(k_1, k_2, t_1+1, t_2)}_{n+1}}{\tau^{(k_1, k_2, t_1, t_2)}_{n+1}\tau^{(k_1, k_2, t_1+1, t_2)}_n},&&
  \tilde e^{(k_1, k_2, t_1, t_2)}_n=\frac{\tau^{(k_1, k_2, t_1, t_2)}_{n+2}\tau^{(k_1, k_2, t_1, t_2+1)}_n}{\tau^{(k_1, k_2, t_1, t_2)}_{n+1}\tau^{(k_1, k_2, t_1, t_2+1)}_{n+1}}.
\end{xxalignat}

By using the semi-infinite bidiagonal matrices
\begin{flalign*}
  &R^{(k_1, k_2, t_1, t_2)}\coloneq\left(q^{(k_1, k_2, t_1, t_2)}_j\delta_{i, j}+\delta_{i+1, j}\right)_{i, j=0}^\infty,&&
  L^{(k_1, k_2, t_1, t_2)}\coloneq\left(e^{(k_1, k_2, t_1, t_2)}_j\delta_{i, j+1}+\delta_{i, j}\right)_{i, j=0}^\infty,\\
  &\tilde R^{(k_1, k_2, t_1, t_2)}\coloneq\left(\tilde q^{(k_1, k_2, t_1, t_2)}_j\delta_{i, j}+\delta_{i+1, j}\right)_{i, j=0}^\infty,&&
  \tilde L^{(k_1, k_2, t_1, t_2)}\coloneq\left(\tilde e^{(k_1, k_2, t_1, t_2)}_j\delta_{i, j+1}+\delta_{i, j}\right)_{i, j=0}^\infty
\end{flalign*}
and the semi-infinite vector
\begin{equation*}
  \bm \phi^{(k_1, k_2, t_1, t_2)}(z)\coloneq
  \begin{pmatrix}
    \phi^{(k_1, k_2, t_1, t_2)}_0(z)\\
    \phi^{(k_1, k_2, t_1, t_2)}_1(z)\\
    \phi^{(k_1, k_2, t_1, t_2)}_2(z)\\
    \vdots
  \end{pmatrix}
\end{equation*}
the relations~\eqref{eq:st-b} are rewritten as
\begin{subequations}\label{eq:st-b-mat}
  \begin{align}
    z\bm\phi^{(k_1+1, k_2, t_1, t_2)}(z)&=R^{(k_1, k_2, t_1, t_2)}\bm \phi^{(k_1, k_2, t_1, t_2)}(z),\label{eq:st-b-mat-1}\\
    \bm\phi^{(k_1, k_2-1, t_1, t_2)}(z)&=L^{(k_1, k_2-1, t_1, t_2)}\bm \phi^{(k_1, k_2, t_1, t_2)}(z),\label{eq:st-b-mat-2}\\
    (z-s_1^{(t_1)})\bm\phi^{(k_1, k_2, t_1+1, t_2)}(z)&=\tilde R^{(k_1, k_2, t_1, t_2)}\bm \phi^{(k_1, k_2, t_1, t_2)}(z),\label{eq:st-b-mat-3}\\
    \bm\phi^{(k_1, k_2, t_1, t_2-1)}(z)&=\tilde L^{(k_1, k_2, t_1, t_2-1)}\bm \phi^{(k_1, k_2, t_1, t_2)}(z).\label{eq:st-b-mat-4}
  \end{align}
\end{subequations}
From \eqref{eq:st-b-mat-1} and \eqref{eq:st-b-mat-2}, we have
\begin{align*}
  z\bm\phi^{(k_1+1, k_2, t_1, t_2)}(z)
  &=L^{(k_1+1, k_2, t_1, t_2)}R^{(k_1, k_2+1, t_1, t_2)}\bm\phi^{(k_1, k_2+1, t_1, t_2)}(z)\\
  &=R^{(k_1, k_2, t_1, t_2)}L^{(k_1, k_2, t_1, t_2)}\bm\phi^{(k_1, k_2+1, t_1, t_2)}(z),
\end{align*}
whose each element gives
\begin{flalign*}
  &\phantom{{}={}}z\phi^{(k_1+1, k_2, t_1, t_2)}_n(z)-\phi^{(k_1, k_2+1, t_1, t_2)}_{n+1}(z)&\\
  &=(q^{(k_1, k_2+1, t_1, t_2)}_n+e^{(k_1+1, k_2, t_1, t_2)}_{n-1})\phi^{(k_1, k_2+1, t_1, t_2)}_n(z)
  +q^{(k_1, k_2+1, t_1, t_2)}_{n-1}e^{(k_1+1, k_2, t_1, t_2)}_{n-1}\phi^{(k_1, k_2+1, t_1, t_2)}_{n-1}(z)&\\
  &=(q^{(k_1, k_2, t_1, t_2)}_n+e^{(k_1, k_2, t_1, t_2)}_n)\phi^{(k_1, k_2+1, t_1, t_2)}_n(z)
  +q^{(k_1, k_2, t_1, t_2)}_n e^{(k_1, k_2, t_1, t_2)}_{n-1}\phi^{(k_1, k_2+1, t_1, t_2)}_{n-1}(z)&
\end{flalign*}
for $n=0, 1, 2, \dots$.
Hence, we obtain
\begin{subequations}\label{eq:d2dtoda}
  \begin{gather}
    q^{(k_1, k_2+1, t_1, t_2)}_n+e^{(k_1+1, k_2, t_1, t_2)}_{n-1}
    =q^{(k_1, k_2, t_1, t_2)}_n+e^{(k_1, k_2, t_1, t_2)}_n,\\
    q^{(k_1, k_2+1, t_1, t_2)}_{n}e^{(k_1+1, k_2, t_1, t_2)}_{n}
    =q^{(k_1, k_2, t_1, t_2)}_{n+1} e^{(k_1, k_2, t_1, t_2)}_{n}
  \end{gather}
with the boundary condition
\begin{equation}\label{eq:bc-d2dtoda}
  e^{(k_1, k_2, t_1, t_2)}_{-1}=0
\end{equation}
for all $k_1, k_2, t_1, t_2 \in \mathbb Z$.
\end{subequations}
The discrete equations~\eqref{eq:d2dtoda} are the compatibility conditions
for the spectral transformations~\eqref{eq:st-b-1} and \eqref{eq:st-b-2}, and called
the discrete two-dimensional Toda lattice.

Similar calculations for each pair of the spectral transformations \eqref{eq:st-b-mat}
also yield
\begin{flalign*}
  z(z-s_1^{(t_1)})\bm\phi^{(k_1+1, k_2, t_1+1, t_2)}(z)
  &=R^{(k_1, k_2, t_1+1, t_2)}\tilde R^{(k_1, k_2, t_1, t_2)}\bm \phi^{(k_1, k_2, t_1, t_2)}(z)&\\
  &=\tilde R^{(k_1+1, k_2, t_1, t_2)}R^{(k_1, k_2, t_1, t_2)}\bm \phi^{(k_1, k_2, t_1, t_2)}(z),&\\
  (z-s_1^{(t_1)})\bm\phi^{(k_1, k_2, t_1+1, t_2)}(z)
  &=L^{(k_1, k_2, t_1+1, t_2)}\tilde R^{(k_1, k_2+1, t_1, t_2)}\bm\phi^{(k_1, k_2+1, t_1, t_2)}(z)&\\
  &=\tilde R^{(k_1, k_2, t_1, t_2)}L^{(k_1, k_2, t_1, t_2)}\bm\phi^{(k_1, k_2+1, t_1, t_2)}(z),&\\
  z\bm\phi^{(k_1+1, k_2, t_1, t_2)}(z)
  &=\tilde L^{(k_1+1, k_2, t_1, t_2)}R^{(k_1, k_2, t_1, t_2+1)}\bm\phi^{(k_1, k_2, t_1, t_2+1)}(z)&\\
  &=R^{(k_1, k_2, t_1, t_2)}\tilde L^{(k_1, k_2, t_1, t_2)}\bm\phi^{(k_1, k_2, t_1, t_2+1)}(z)&\\
  \bm\phi^{(k_1, k_2, t_1, t_2)}(z)
  &=\tilde L^{(k_1, k_2, t_1, t_2)}L^{(k_1, k_2, t_1, t_2+1)}\bm\phi^{(k_1, k_2+1, t_1, t_2+1)}(z)&\\
  &=L^{(k_1, k_2, t_1, t_2)}\tilde L^{(k_1, k_2+1, t_1, t_2)}\bm\phi^{(k_1, k_2+1, t_1, t_2+1)}(z),&\\
  (z-s_1^{(t_1)})\bm\phi^{(k_1, k_2, t_1+1, t_2)}(z)
  &=\tilde L^{(k_1, k_2, t_1+1, t_2)}\tilde R^{(k_1, k_2, t_1, t_2+1)}\bm\phi^{(k_1, k_2, t_1, t_2+1)}(z)&\\
  &=\tilde R^{(k_1, k_2, t_1, t_2)}\tilde L^{(k_1, k_2, t_1, t_2)}\bm\phi^{(k_1, k_2, t_1, t_2+1)}(z),&
\end{flalign*}
whose elements give the relations
\newlength\tempalength
\tempalength\belowdisplayskip
\belowdisplayskip-12pt plus 0pt minus 0pt
\begin{subequations}\label{eq:nd2d-toda-1}
  \begin{gather}
    q^{(k_1, k_2, t_1+1, t_2)}_{n}+\tilde q^{(k_1, k_2, t_1, t_2)}_{n+1}
    =q^{(k_1, k_2, t_1, t_2)}_{n+1}+\tilde q^{(k_1+1, k_2, t_1, t_2)}_{n},\\
    q^{(k_1, k_2, t_1+1, t_2)}_n\tilde q^{(k_1, k_2, t_1, t_2)}_n
    =q^{(k_1, k_2, t_1, t_2)}_n\tilde q^{(k_1+1, k_2, t_1, t_2)}_n,
  \end{gather}
\end{subequations}
\begin{subequations}\label{eq:nd2d-toda-2}
  \begin{gather}
    \tilde q^{(k_1, k_2+1, t_1, t_2)}_n+e^{(k_1, k_2, t_1+1, t_2)}_{n-1}
    =\tilde q^{(k_1, k_2, t_1, t_2)}_n+e^{(k_1, k_2, t_1, t_2)}_n,\\
    \tilde q^{(k_1, k_2+1, t_1, t_2)}_n e^{(k_1, k_2, t_1+1, t_2)}_n
    =\tilde q^{(k_1, k_2, t_1, t_2)}_{n+1} e^{(k_1, k_2, t_1, t_2)}_n,
  \end{gather}
\end{subequations}
\begin{subequations}\label{eq:nd2d-toda-3}
  \begin{gather}
    q^{(k_1, k_2, t_1, t_2+1)}_n+\tilde e^{(k_1+1, k_2, t_1, t_2)}_{n-1}
    =q^{(k_1, k_2, t_1, t_2)}_n+\tilde e^{(k_1, k_2, t_1, t_2)}_n,\\
    q^{(k_1, k_2, t_1, t_2+1)}_n \tilde e^{(k_1+1, k_2, t_1, t_2)}_n
    =q^{(k_1, k_2, t_1, t_2)}_{n+1} \tilde e^{(k_1, k_2, t_1, t_2)}_n,
  \end{gather}
\end{subequations}
\begin{subequations}\label{eq:nd2d-toda-4}
  \begin{gather}
    e^{(k_1, k_2, t_1, t_2+1)}_n+\tilde e^{(k_1, k_2, t_1, t_2)}_n
    =e^{(k_1, k_2, t_1, t_2)}_n+\tilde e^{(k_1, k_2+1, t_1, t_2)}_n,\\
    e^{(k_1, k_2, t_1, t_2+1)}_n \tilde e^{(k_1, k_2, t_1, t_2)}_{n+1}
    =e^{(k_1, k_2, t_1, t_2)}_{n+1} \tilde e^{(k_1, k_2+1, t_1, t_2)}_{n},
  \end{gather}
\end{subequations}
\begin{subequations}\label{eq:nd2d-toda-5}
  \begin{gather}
    \tilde q^{(k_1, k_2, t_1, t_2+1)}_n+\tilde e^{(k_1, k_2, t_1+1, t_2)}_{n-1}
    =\tilde q^{(k_1, k_2, t_1, t_2)}_n+\tilde e^{(k_1, k_2, t_1, t_2)}_n,\\
    \tilde q^{(k_1, k_2, t_1, t_2+1)}_{n}\tilde e^{(k_1, k_2, t_1+1, t_2)}_{n}
    =\tilde q^{(k_1, k_2, t_1, t_2)}_{n+1} \tilde e^{(k_1, k_2, t_1, t_2)}_{n}
    \global\belowdisplayskip\tempalength
  \end{gather}
\end{subequations}
for $n=0, 1, 2, \dots$ with the boundary condition~\eqref{eq:bc-d2dtoda} and
\begin{equation*}
  \tilde e^{(k_1, k_2, t_1, t_2)}_{-1}=0
\end{equation*}
for all $k_1, k_2, t_1, t_2 \in \mathbb Z$.

In these discrete equations,
the parameters $s_1^{(t_1)}$ and $s_2^{(t_2)}$ do not appear explicitly.
The parameters are, in fact, embedded into boundary conditions as follows.

\paragraph{For equations \eqref{eq:nd2d-toda-1}}
Subtraction of \eqref{eq:st-b-1} from \eqref{eq:st-b-3} yields the relation
\begin{equation}\label{eq:st-b-5}
  \fl
  (z-s_1^{(t_1)})\phi^{(k_1, k_2, t_1+1, t_2)}_n(z)=z\phi^{(k_1+1, k_2, t_1, t_2)}_n(z)+a^{(k_1, k_2, t_1, t_2)}_n\phi^{(k_1, k_2, t_1, t_2)}_n(z),
\end{equation}
where
\begin{equation*}
  a^{(k_1, k_2, t_1, t_2)}_n\coloneq\tilde q^{(k_1, k_2, t_1, t_2)}_n-q^{(k_1, k_2, t_1, t_2)}_n.
\end{equation*}
The relation~\eqref{eq:st-b-5} induces
\begin{align*}
  a^{(k_1, k_2, t_1, t_2)}_n
  &=-s_1^{(t_1)}\frac{\phi^{(k_1, k_2, t_1+1, t_2)}_n(0)}{\phi^{(k_1, k_2, t_1, t_2)}_n(0)}
  =-s_1^{(t_1)}\frac{\phi^{(k_1+1, k_2, t_1, t_2)}_n(s_1^{(t_1)})}{\phi^{(k_1, k_2, t_1, t_2)}_n(s_1^{(t_1)})}\\
  &=-s_1^{(t_1)}\frac{\tau^{(k_1, k_2, t_1, t_2)}_n\tau^{(k_1+1, k_2, t_1+1, t_2)}_n}{\tau^{(k_1+1, k_2, t_1, t_2)}_n\tau^{(k_1, k_2, t_1+1, t_2)}_n}.
\end{align*}
By using the variable $a^{(k_1, k_2, t_1, t_2)}_n$, equations~\eqref{eq:nd2d-toda-1}
are rewritten as
\begin{subequations}\label{eq:nd2d-toda-1-sf}
  \begin{gather}
    \tilde q^{(k_1, k_2, t_1, t_2)}_n=q^{(k_1, k_2, t_1, t_2)}_n+a^{(k_1, k_2, t_1, t_2)}_n,\label{eq:nd2d-toda-1-sf-1}\\
    q^{(k_1, k_2, t_1+1, t_2)}_n\tilde q^{(k_1, k_2, t_1, t_2)}_n
    =q^{(k_1, k_2, t_1, t_2)}_n\tilde q^{(k_1+1, k_2, t_1, t_2)}_n,\label{eq:nd2d-toda-1-sf-2}\\
    a^{(k_1, k_2, t_1, t_2)}_{n+1}\tilde q^{(k_1, k_2, t_1, t_2)}_n
    =a^{(k_1, k_2, t_1, t_2)}_n\tilde q^{(k_1+1, k_2, t_1, t_2)}_n,\label{eq:nd2d-toda-1-sf-3}
  \end{gather}
  for $n=0, 1, 2, \dots$ with the boundary condition
  \begin{equation}
    a^{(k_1, k_2, t_1, t_2)}_0=-s_1^{(t_1)}
  \end{equation}
  for all $k_1, k_2, t_1, t_2 \in \mathbb Z$.
\end{subequations}
Note that equation \eqref{eq:nd2d-toda-1-sf-1} is readily transformed
into the bilinear equation
\begin{equation}\label{eq:nd2d-toda-1-bilinear}
  \fl
  \tau^{(k_1, k_2, t_1+1, t_2)}_{n+1}\tau^{(k_1+1, k_2, t_1, t_2)}_n
  =\tau^{(k_1, k_2, t_1+1, t_2)}_n\tau^{(k_1+1, k_2, t_1, t_2)}_{n+1}
  -s_1^{(t_1)}\tau^{(k_1, k_2, t_1, t_2)}_{n+1}\tau^{(k_1+1, k_2, t_1+1, t_2)}_n
\end{equation}
and equations \eqref{eq:nd2d-toda-1-sf-2} and \eqref{eq:nd2d-toda-1-sf-3} are
obvious identical equations of $\tau^{(k_1, k_2, t_1, t_2)}_n$.

\paragraph{For equations \eqref{eq:nd2d-toda-2}}
Subtraction of \eqref{eq:st-b-2} from \eqref{eq:st-b-3} yields the relation
\begin{equation}\label{eq:st-b-6}
  \fl
  (z-s_1^{(t_1)})\phi^{(k_1, k_2, t_1+1, t_2)}_n(z)
  =\phi^{(k_1, k_2-1, t_1, t_2)}_{n+1}(z)+b^{(k_1, k_2-1, t_1, t_2)}_n \phi^{(k_1, k_2, t_1, t_2)}_n(z),
\end{equation}
where
\begin{equation*}
  b^{(k_1, k_2, t_1, t_2)}_n\coloneq\tilde q^{(k_1, k_2+1, t_1, t_2)}_n-e^{(k_1, k_2, t_1, t_2)}_n.
\end{equation*}
The relation \eqref{eq:st-b-6} induces
\begin{equation*}
  b^{(k_1, k_2, t_1, t_2)}_n
  =-\frac{\phi^{(k_1, k_2, t_1, t_2)}_{n+1}(s_1^{(t_1)})}{\phi^{(k_1, k_2+1, t_1, t_2)}_n(s_1^{(t_1)})}
  =\frac{\tau^{(k_1, k_2+1, t_1, t_2)}_n \tau^{(k_1, k_2, t_1+1, t_2)}_{n+1}}{\tau^{(k_1, k_2, t_1, t_2)}_{n+1} \tau^{(k_1, k_2+1, t_1+1, t_2)}_n}.
\end{equation*}
By using the variable $b^{(k_1, k_2, t_1, t_2)}_n$, equations~\eqref{eq:nd2d-toda-2}
are rewritten as
\begin{subequations}\label{eq:nd2d-toda-2-sf}
  \begin{gather}
    \tilde q^{(k_1, k_2+1, t_1, t_2)}_n=b^{(k_1, k_2, t_1, t_2)}_n+e^{(k_1, k_2, t_1, t_2)}_n,\label{eq:nd2d-toda-2-sf-1}\\
    \tilde q^{(k_1, k_2+1, t_1, t_2)}_n b^{(k_1, k_2, t_1, t_2)}_{n+1}
    =\tilde q^{(k_1, k_2, t_1, t_2)}_{n+1} b^{(k_1, k_2, t_1, t_2)}_n,\\
    \tilde q^{(k_1, k_2+1, t_1, t_2)}_n e^{(k_1, k_2, t_1+1, t_2)}_{n}
    =\tilde q^{(k_1, k_2, t_1, t_2)}_{n+1} e^{(k_1, k_2, t_1, t_2)}_n
  \end{gather}
  for $n=0, 1, 2, \dots$ with the boundary condition
  \begin{equation}
    b^{(k_1, k_2, t_1, t_2)}_0=q^{(k_1, k_2, t_1, t_2)}_0-s_1^{(t_1)}=\tilde q^{(k_1, k_2, t_1, t_2)}_0
  \end{equation}
  for all $k_1, k_2, t_1, t_2 \in \mathbb Z$.
\end{subequations}
Equation~\eqref{eq:nd2d-toda-2-sf-1} is transformed into the bilinear equation
\begin{equation*}
  \fl
  \tau^{(k_1, k_2, t_1, t_2)}_{n+1}\tau^{(k_1, k_2+1, t_1+1, t_2)}_{n+1}
  =\tau^{(k_1, k_2, t_1+1, t_2)}_{n+1}\tau^{(k_1, k_2+1, t_1, t_2)}_{n+1}
  +\tau^{(k_1, k_2, t_1, t_2)}_{n+2}\tau^{(k_1, k_2+1, t_1+1, t_2)}_n.
\end{equation*}

\paragraph{For equations \eqref{eq:nd2d-toda-3}}
Subtraction of \eqref{eq:st-b-4} from \eqref{eq:st-b-1} yields the relation
\begin{equation}\label{eq:st-b-7}
  z\phi^{(k_1+1, k_2, t_1, t_2)}_n(z)
  =\phi^{(k_1, k_2, t_1, t_2-1)}_{n+1}(z)+d^{(k_1, k_2, t_1, t_2-1)}_n \phi^{(k_1, k_2, t_1, t_2)}_n(z),
\end{equation}
where
\begin{equation*}
  d^{(k_1, k_2, t_1, t_2)}_n\coloneq q^{(k_1, k_2, t_1, t_2+1)}_n-\tilde e^{(k_1, k_2, t_1, t_2)}_n.
\end{equation*}
The relation \eqref{eq:st-b-7} induces
\begin{equation*}
  d^{(k_1, k_2, t_1, t_2)}_n
  =-\frac{\phi^{(k_1, k_2, t_1, t_2)}_{n+1}(0)}{\phi^{(k_1, k_2, t_1, t_2+1)}_n(0)}
  =\frac{\tau^{(k_1, k_2, t_1, t_2+1)}_n\tau^{(k_1+1, k_2, t_1, t_2)}_{n+1}}{\tau^{(k_1, k_2, t_1, t_2)}_{n+1}\tau^{(k_1+1, k_2, t_1, t_2+1)}_n}.
\end{equation*}
By using the variable $d^{(k_1, k_2, t_1, t_2)}_n$, equations~\eqref{eq:nd2d-toda-3}
are rewritten as
\begin{subequations}\label{eq:nd2d-toda-3-sf}
  \begin{gather}
    q^{(k_1, k_2, t_1, t_2+1)}_n=d^{(k_1, k_2, t_1, t_2)}_n+\tilde e^{(k_1, k_2, t_1, t_2)}_n,\label{eq:nd2d-toda-3-sf-1}\\
    q^{(k_1, k_2, t_1, t_2+1)}_n\tilde e^{(k_1+1, k_2, t_1, t_2)}_n
    =q^{(k_1, k_2, t_1, t_2)}_{n+1}\tilde e^{(k_1, k_2, t_1, t_2)}_n,\\
    d^{(k_1, k_2, t_1, t_2)}_{n+1}\tilde e^{(k_1, k_2, t_1, t_2)}_n
    =d^{(k_1, k_2, t_1, t_2)}_{n}\tilde e^{(k_1+1, k_2, t_1, t_2)}_n
  \end{gather}
  for $n=0, 1, 2, \dots$ with the boundary condition
  \begin{equation}
    d^{(k_1, k_2, t_1, t_2)}_0=q^{(k_1, k_2, t_1, t_2)}_0=\tilde q^{(k_1, k_2, t_1, t_2)}_0+s_1^{(t_1)}
  \end{equation}
  for all $k_1, k_2, t_1, t_2 \in \mathbb Z$.
\end{subequations}
Equation~\eqref{eq:nd2d-toda-3-sf-1} is transformed into the bilinear equation
\begin{equation*}
  \fl
  \tau^{(k_1, k_2, t_1, t_2)}_{n+1}\tau^{(k_1+1, k_2, t_1, t_2+1)}_{n+1}
  =\tau^{(k_1, k_2, t_1, t_2+1)}_{n+1}\tau^{(k_1+1, k_2, t_1, t_2)}_{n+1}
  +\tau^{(k_1, k_2, t_1, t_2)}_{n+2}\tau^{(k_1+1, k_2, t_1, t_2+1)}_n.
\end{equation*}

\paragraph{For equations \eqref{eq:nd2d-toda-4}}
Subtraction of \eqref{eq:st-b-4} from \eqref{eq:st-b-2} yields the relation
\begin{equation}\label{eq:st-b-8}
  \phi^{(k_1, k_2-1, t_1, t_2)}_{n+1}(z)
  =\phi^{(k_1, k_2, t_1, t_2-1)}_{n+1}(z)+f^{(k_1, k_2-1, t_1, t_2-1)}_n \phi^{(k_1, k_2, t_1, t_2)}_n(z),
\end{equation}
where
\begin{equation*}
  f^{(k_1, k_2, t_1, t_2)}_n\coloneq e^{(k_1, k_2, t_1, t_2+1)}_n-\tilde e^{(k_1, k_2+1, t_1, t_2)}_n.
\end{equation*}
Since the same discussion for another one of the pair of the monic biorthogonal polynomial sequences
$\{\psi^{(k_1, k_2, t_1, t_2)}_n\}_{n=0}^\infty$ leads us to
the ``dual'' version of the bilinear equation~\eqref{eq:nd2d-toda-1-bilinear}
\begin{equation}\label{eq:nd2d-toda-1-bilinear-dual}
  \fl
  \tau^{(k_1, k_2, t_1, t_2+1)}_{n+1}\tau^{(k_1, k_2+1, t_1, t_2)}_n
  =\tau^{(k_1, k_2, t_1, t_2+1)}_n\tau^{(k_1, k_2+1, t_1, t_2)}_{n+1}
  -s_2^{(t_2)}\tau^{(k_1, k_2, t_1, t_2)}_{n+1}\tau^{(k_1, k_2+1, t_1, t_2+1)}_n,
\end{equation}
the relation~\eqref{eq:st-b-8} induces
\begin{flalign*}
  &\phantom{{}={}}f^{(k_1, k_2, t_1, t_2)}_n&\\
  &=\frac{\mathcal B^{(k_1, k_2, t_1, t_2+1)}[\phi^{(k_1, k_2, t_1, t_2+1)}_{n+1}(z), z^{n+1}]-\mathcal B^{(k_1, k_2+1, t_1, t_2)}[\phi^{(k_1, k_2+1, t_1, t_2)}_{n+1}(z), z^{n}(z-s_2^{(t_2)})]}{\mathcal B^{(k_1, k_2+1, t_1, t_2+1)}[\phi^{(k_1, k_2+1, t_1, t_2+1)}_n(z), z^n]}&\\
  &=\frac{(\tau^{(k_1, k_2, t_1, t_2+1)}_{n+2}\tau^{(k_1, k_2+1, t_1, t_2)}_{n+1}-\tau^{(k_1, k_2, t_1, t_2+1)}_{n+1}\tau^{(k_1, k_2+1, t_1, t_2)}_{n+2})\tau^{(k_1, k_2+1, t_1, t_2+1)}_{n}}{\tau^{(k_1, k_2, t_1, t_2+1)}_{n+1}\tau^{(k_1, k_2+1, t_1, t_2)}_{n+1}\tau^{(k_1, k_2+1, t_1, t_2+1)}_{n+1}}&\\
  &=-s_2^{(t_2)} \frac{\tau^{(k_1, k_2, t_1, t_2)}_{n+2}\tau^{(k_1, k_2+1, t_1, t_2+1)}_{n}}{\tau^{(k_1, k_2, t_1, t_2+1)}_{n+1}\tau^{(k_1, k_2+1, t_1, t_2)}_{n+1}}.&
\end{flalign*}
By using the variable $f^{(k_1, k_2, t_1, t_2)}_n$, equations~\eqref{eq:nd2d-toda-4}
are rewritten as
\begin{subequations}\label{eq:nd2d-toda-4-sf}
  \begin{gather}
    e^{(k_1, k_2, t_1, t_2+1)}_n=f^{(k_1, k_2, t_1, t_2)}_n+\tilde e^{(k_1, k_2+1, t_1, t_2)}_n,\label{eq:nd2d-toda-4-sf-1}\\
    f^{(k_1, k_2, t_1, t_2)}_{n+1}\tilde e^{(k_1, k_2+1, t_1, t_2)}_n
    =f^{(k_1, k_2, t_1, t_2)}_{n}\tilde e^{(k_1, k_2, t_1, t_2)}_{n+1},\\
    e^{(k_1, k_2, t_1, t_2+1)}_n \tilde e^{(k_1, k_2, t_1, t_2)}_{n+1}
    =e^{(k_1, k_2, t_1, t_2)}_{n+1} \tilde e^{(k_1, k_2+1, t_1, t_2)}_{n}
  \end{gather}
  for $n=0, 1, 2, \dots$ with the boundary condition
  \begin{flalign}
    f^{(k_1, k_2, t_1, t_2)}_0
    &=-s_2^{(t_2)}e^{(k_1, k_2, t_1, t_2)}_0 \frac{\tau^{(k_1, k_2, t_1, t_2)}_1}{\tau^{(k_1, k_2, t_1, t_2+1)}_1}
    =-s_2^{(t_2)}e^{(k_1, k_2, t_1, t_2)}_0 \frac{\mu^{(t_1, t_2)}_{k_1, k_2}}{\mu^{(t_1, t_2)}_{k_1, k_2+1}-s_2^{(t_2)}\mu^{(t_1, t_2)}_{k_1, k_2}}&\nonumber\\
    &=\frac{-s_2^{(t_2)}e^{(k_1, k_2, t_1, t_2)}_0 }{\tau^{(k_1, k_2+1, t_1, t_2)}_1/\tau^{(k_1, k_2, t_1, t_2)}_1-s_2^{(t_2)}}&
  \end{flalign}
\end{subequations}
for all $k_1$, $k_2$, $t_1$ and $t_2$.
Equation~\eqref{eq:nd2d-toda-4-sf-1} is transformed into
the bilinear equation~\eqref{eq:nd2d-toda-1-bilinear-dual}.

\paragraph{For equations \eqref{eq:nd2d-toda-5}}
Subtraction of \eqref{eq:st-b-4} from \eqref{eq:st-b-3} yields the relations
\begin{equation}\label{eq:st-b-9}
  \fl
  (z-s_1^{(t_1)})\phi^{(k_1, k_2, t_1+1, t_2)}_n(z)
  =\phi^{(k_1, k_2, t_1, t_2-1)}_{n+1}(z)+g^{(k_1, k_2, t_1, t_2-1)}_n\phi^{(k_1, k_2, t_1, t_2)}_n(z),
\end{equation}
where
\begin{equation*}
  g^{(k_1, k_2, t_1, t_2)}_n\coloneq\tilde q^{(k_1, k_2, t_1, t_2+1)}_n-\tilde e^{(k_1, k_2, t_1, t_2)}_n.
\end{equation*}
The relation~\eqref{eq:st-b-9} induces
\begin{equation*}
  g^{(k_1, k_2, t_1, t_2)}_n
  =-\frac{\phi^{(k_1, k_2, t_1, t_2)}_{n+1}(s_1^{(t_1)})}{\phi^{(k_1, k_2, t_1, t_2+1)}_{n}(s_1^{(t_1)})}
  =\frac{\tau^{(k_1, k_2, t_1, t_2+1)}_{n}\tau^{(k_1, k_2, t_1+1, t_2)}_{n+1}}{\tau^{(k_1, k_2, t_1, t_2)}_{n+1}\tau^{(k_1, k_2, t_1+1, t_2+1)}_{n}}.
\end{equation*}
By using the variable $g^{(k_1, k_2, t_1, t_2)}_n$, equations~\eqref{eq:nd2d-toda-5}
are rewritten as
\begin{subequations}\label{eq:nd2d-toda-5-sf}
  \begin{gather}
    \tilde q^{(k_1, k_2, t_1, t_2+1)}_n=g^{(k_1, k_2, t_1, t_2)}_n+\tilde e^{(k_1, k_2, t_1, t_2)}_n,\label{eq:nd2d-toda-5-sf-1}\\
    \tilde q^{(k_1, k_2, t_1, t_2+1)}_n g^{(k_1, k_2, t_1, t_2)}_{n+1}
    =\tilde q^{(k_1, k_2, t_1, t_2)}_{n+1} g^{(k_1, k_2, t_1, t_2)}_n,\\
    \tilde q^{(k_1, k_2, t_1, t_2+1)}_n\tilde e^{(k_1, k_2, t_1+1, t_2)}_n
    =\tilde q^{(k_1, k_2, t_1, t_2)}_{n+1}\tilde e^{(k_1, k_2, t_1, t_2)}_n
  \end{gather}
  for $n=0, 1, 2, \dots$ with the boundary condition
  \begin{equation}
    g^{(k_1, k_2, t_1, t_2)}_0=q^{(k_1, k_2, t_1, t_2)}_0-s_1^{(t_1)}=\tilde q^{(k_1, k_2, t_1, t_2)}_0
  \end{equation}
  for all $k_1, k_2, t_1, t_2 \in \mathbb Z$.
\end{subequations}
Equation \eqref{eq:nd2d-toda-5-sf-1} is transformed into
the bilinear equation
\begin{equation*}
  \fl
  \tau^{(k_1, k_2, t_1, t_2)}_{n+1}\tau^{(k_1, k_2, t_1+1, t_2+1)}_{n+1}
  =\tau^{(k_1, k_2, t_1, t_2+1)}_{n+1}\tau^{(k_1, k_2, t_1+1, t_2)}_{n+1}
  +\tau^{(k_1, k_2, t_1, t_2)}_{n+2}\tau^{(k_1, k_2, t_1+1, t_2+1)}_n.
\end{equation*}

In this paper, we call the system of \eqref{eq:d2dtoda},
\eqref{eq:nd2d-toda-1-sf}, \eqref{eq:nd2d-toda-2-sf},
\eqref{eq:nd2d-toda-3-sf}, \eqref{eq:nd2d-toda-4-sf} and
\eqref{eq:nd2d-toda-5-sf} the nd-2D-Toda lattice.

\section{\boldmath $(M, 1)$-reduction}\label{sec:m-1-reduction}

In this section, we consider a special case of the chain of
the monic biorthogonal polynomials:
the bilinear forms satisfy the condition $\mathcal B^{(k_1+M, k_2, t_1, t_2)}=\mathcal B^{(k_1, k_2+1, t_1, t_2)}$ for all $k_1$, $k_2$, $t_1$ and $t_2$, where $M$ is a positive integer.
The condition is equivalent to
\begin{equation*}
  \mathcal B^{(k_1, k_2, t_1, t_2)}[z^{i+M}, z^{j}]=\mathcal B^{(k_1, k_2, t_1, t_2)}[z^{i}, z^{j+1}],\quad
  i, j=0, 1, 2, \dots.
\end{equation*}

\subsection{Nonautonomous discrete semi-infinite hungry Toda lattice}
Introduce a new linear functional
$\mathcal L^{(k_1, k_2, t_1, t_2)}\colon \mathbb C[z] \to \mathbb C$ by
\begin{equation*}
  \mathcal L^{(k_1, k_2, t_1, t_2)}[z^i]\coloneq\mathcal B^{(k_1, k_2, t_1, t_2)}[z^i, 1],\quad
  i=0, 1, 2, \dots.
\end{equation*}
Since the relation
\begin{equation*}
  \mathcal L^{(k_1, k_2, t_1, t_2)}[z^{i+Mj}]
  =\mathcal B^{(k_1, k_2, t_1, t_2)}[z^{i+Mj}, 1]
  =\mathcal B^{(k_1, k_2, t_1, t_2)}[z^{i}, z^j]
\end{equation*}
holds,
the pair of the monic biorthogonal polynomial sequences $\{\phi^{(k_1, k_2, t_1, t_2)}_n(z)\}_{n=0}^\infty$
and $\{\psi^{(k_1, k_2, t_1, t_2)}_n(z)\}_{n=0}^\infty$ with respect to $\mathcal B^{(k_1, k_2, t_1, t_2)}$
satisfy ``$(M, 1)$-biorthogonal relation'' with respect to $\mathcal L^{(k_1, k_2, t_1, t_2)}$:
\begin{multline*}
  \mathcal L^{(k_1, k_2, t_1, t_2)}[\phi^{(k_1, k_2, t_1, t_2)}_m(z)\psi^{(k_1, k_2, t_1, t_2)}_n(z^M)]
  =\frac{\tau^{(k_1, k_2, t_1, t_2)}_{n+1}}{\tau^{(k_1, k_2, t_1, t_2)}_n}\delta_{m, n},\\
  m, n=0, 1, 2, \dots.
\end{multline*}

Hereafter, we will fix $k_2$ and $t_1$ to zero
and consider only the time variables $k_1$ and $t_2$;
we will simply write $\mathcal L^{(k, t)}$ and $\phi^{(k, t)}_n(z)$ instead of
$\mathcal L^{(k, 0, 0, t)}$ and $\phi^{(k, 0, 0, t)}_n(z)$, respectively.
We will also omit $k_2$ and $t_1$ for all the other variables in the same manner.
Then, the determinant representation of $\phi^{(k, t)}_n(z)$ is given by
\begin{multline*}
  \fl
  \phi^{(k, t)}_0(z)=1,\quad
  \phi^{(k, t)}_n(z)=\frac{1}{\tau^{(k, t)}_n}
  \begin{vmatrix}
    \mu^{(t)}_{k} & \mu^{(t)}_{k+M} & \dots & \mu^{(t)}_{k+M(n-1)} & 1\\
    \mu^{(t)}_{k+1} & \mu^{(t)}_{k+1+M} & \dots & \mu^{(t)}_{k+1+M(n-1)} & z\\
    \vdots & \vdots & & \vdots & \vdots\\
    \mu^{(t)}_{k+n-1} & \mu^{(t)}_{k+n-1+M} & \dots & \mu^{(t)}_{k+n-1+M(n-1)} & z^{n-1}\\
    \mu^{(t)}_{k+n} & \mu^{(t)}_{k+n+M} & \dots & \mu^{(t)}_{k+n+M(n-1)} & z^{n}
  \end{vmatrix},\\
  n=1, 2, 3, \dots
\end{multline*}
where
\begin{alignat}{2}
  &\mu^{(t)}_m\coloneq\mathcal{L}^{(0, t)}[z^m],&\quad& m=0, 1, 2, \dots,\nonumber\\
  &\tau^{(k, t)}_0\coloneq1,\quad \tau^{(k, t)}_n\coloneq|\mu^{(t)}_{k+i+Mj}|_{i, j=0}^{n-1},&& n=1, 2, 3, \dots.\label{eq:m1-tau-def}
\end{alignat}
Note that the moment $\mu^{(t)}_m$ satisfies the relations
\begin{equation*}
  \mathcal L^{(k, t)}[z^m]=\mu^{(t)}_{k+m},\quad
  \mu^{(t+1)}_{m}=\mu^{(t)}_{m+M}-s^{(t)}\mu^{(t)}_m,
\end{equation*}
where we simply write $s^{(t)}$ instead of $s_2^{(t)}$.

From the discussion in \sref{sec:biorth-polyn},
the monic $(M, 1)$-biorthogonal polynomials $\{\phi^{(k, t)}_n(z)\}_{n=0}^\infty$ satisfy
the following relations
\begin{subequations}\label{eq:st-m1}
  \begin{align}
    z\phi^{(k+1, t)}_n(z)
    &=\phi^{(k, t)}_{n+1}(z)+q^{(k, t)}_n\phi^{(k, t)}_n(z),\label{eq:st-m1-1}\\
    \phi^{(k-M, t)}_{n+1}(z)
    &=\phi^{(k, t)}_{n+1}(z)+e^{(k-M, t)}_n\phi^{(k, t)}_n(z),\label{eq:st-m1-2}\\
    \phi^{(k, t-1)}_{n+1}(z)
    &=\phi^{(k, t)}_{n+1}(z)+\tilde e^{(k, t-1)}_n \phi^{(k, t)}_n(z),\label{eq:st-m1-3}\\
    z\phi^{(k+1, t)}_n(z)
    &=\phi^{(k, t-1)}_{n+1}(z)+d^{(k, t-1)}_n \phi^{(k, t)}_n(z),\\
    \phi^{(k-M, t)}_{n+1}(z)
    &=\phi^{(k, t-1)}_{n+1}(z)+f^{(k-M, t-1)}_n \phi^{(k, t)}_n(z),
  \end{align}
\end{subequations}
where
\begin{subequations}\label{eq:ndhtoda-tau}
\begin{xxalignat}{7}
  &q^{(k, t)}_n=\frac{\tau^{(k, t)}_n \tau^{(k+1, t)}_{n+1}}{\tau^{(k, t)}_{n+1}\tau^{(k+1, t)}_n},&\quad&
  e^{(k, t)}_n=\frac{\tau^{(k, t)}_{n+2} \tau^{(k+M, t)}_n}{\tau^{(k, t)}_{n+1}\tau^{(k+M, t)}_{n+1}},&&
  \tilde e^{(k, t)}_n=\frac{\tau^{(k, t)}_{n+2} \tau^{(k, t+1)}_n}{\tau^{(k, t)}_{n+1}\tau^{(k, t+1)}_{n+1}},\\
  &d^{(k, t)}_n=\frac{\tau^{(k, t+1)}_n \tau^{(k+1, t)}_{n+1}}{\tau^{(k, t)}_{n+1}\tau^{(k+1, t+1)}_n},&\quad&
  f^{(k, t)}_n=-s^{(t)}\frac{\tau^{(k, t)}_{n+2}\tau^{(k+M, t+1)}_n}{\tau^{(k, t+1)}_{n+1}\tau^{(k+M, t)}_{n+1}}.
\end{xxalignat}
\end{subequations}
We omitted the variable $\tilde q^{(k_1, k_2, t_1, t_2)}_n$ and
its related relations and variables,
because we will not use them in the subsequent discussion.
The compatibility conditions for \eqref{eq:st-m1} give the recurrence relations
\begin{subequations}\label{eq:ndhtoda-sf}
\begin{alignat}{2}
  &q^{(k, t+1)}_n=d^{(k, t)}_n+\tilde e^{(k, t)}_n,&\quad
  &e^{(k, t+1)}_n=f^{(k, t)}_n+\tilde e^{(k+M, t)}_n,\\
  &d^{(k, t)}_{n+1}=d^{(k, t)}_n\frac{q^{(k, t)}_{n+1}}{q^{(k, t+1)}_n},&\quad
  &f^{(k, t)}_{n+1}=f^{(k, t)}_n\frac{e^{(k, t)}_{n+1}}{e^{(k, t+1)}_n},\\
  &\tilde e^{(k+1, t)}_n=\tilde e^{(k, t)}_n\frac{q^{(k, t)}_{n+1}}{q^{(k, t+1)}_n},&\quad
  &\tilde e^{(k, t)}_{n+1}=\tilde e^{(k+M, t)}_n\frac{e^{(k, t)}_{n+1}}{e^{(k, t+1)}_n}
\end{alignat}
for $n=0, 1, 2, \dots$ with the boundary condition
\begin{equation}\label{eq:ndhtoda-sf-bc}
  \fl
  d^{(k, t)}_0=q^{(k, t)}_0,\quad
  f^{(k, t)}_0=\frac{-e^{(k, t)}_0 s^{(t)}}{\prod_{j=0}^{M-1}q^{(k+j, t)}_0-s^{(t)}},\quad
  \tilde e^{(k, t)}_0=\frac{e^{(k, t)}_0\prod_{j=0}^{M-1}q^{(k+j, t)}_0}{\prod_{j=0}^{M-1}q^{(k+j, t)}_0-s^{(t)}}
\end{equation}
for all $k, t \in \mathbb Z$.
\end{subequations}

We should remark that, if $s^{(t)}=0$ for all $t$,
then $f^{(k, t)}_n=0$, $e^{(k, t)}_n=\tilde e^{(k, t)}_n$
and $\mathcal{L}^{(k, t+1)}=\mathcal{L}^{(k+M, t)}$
hold for all $k$, $t$ and $n$.
Therefore, equations~\eqref{eq:ndhtoda-sf} are reduced to
\begin{gather*}
  q^{(k+M)}_n=d^{(k)}_n+e^{(k)}_n,\quad
  e^{(k+1)}_n=e^{(k)}_n\frac{q^{(k)}_{n+1}}{q^{(k+M)}_n},\quad
  d^{(k)}_{n+1}=d^{(k)}_n\frac{q^{(k)}_{n+1}}{q^{(k+M)}_n},
\end{gather*}
where we omitted the time variable $t$.
Elimination of $d^{(k)}_n$ yields
\begin{equation}\label{eq:dhtoda}
  q^{(k+M)}_n+e^{(k+1)}_{n-1}=q^{(k)}_{n}+e^{(k)}_{n},\quad
  q^{(k+M)}_n e^{(k+1)}_n=q^{(k)}_{n+1} e^{(k)}_n.
\end{equation}
The system~\eqref{eq:dhtoda} is called the discrete hungry Toda lattice,
which is the reason why we call
the system~\eqref{eq:ndhtoda-sf} the ndh-Toda lattice.

\subsection{Nonautonomous discrete finite hungry Toda lattice}
In this subsection, we consider
the ndh-Toda lattice~\eqref{eq:ndhtoda-sf} with
the finite lattice boundary condition
\begin{equation}\label{eq:finite-lattice-BC}
  \tau^{(k, t)}_n=0\quad\text{if $n>N$}
\end{equation}
for all $k, t \in \mathbb Z$,
where the lattice size $N$ is a positive integer.
By imposing the boundary condition,
the pair of the semi-infinite biorthogonal polynomial sequences
$\{\phi^{(k, t)}_n(z)\}_{n=0}^\infty$ and
$\{\psi^{(k, t)}_n(z)\}_{n=0}^\infty$
are reduced to the pair of finite polynomial sequences
$\{\phi^{(k, t)}_n(z)\}_{n=0}^N$ and
$\{\psi^{(k, t)}_n(z)\}_{n=0}^N$.
Since
\begin{equation}\label{eq:ndhtoda-finite-cond}
  e^{(k, t)}_{N-1}=\tilde e^{(k, t)}_{N-1}=0
\end{equation}
holds from \eqref{eq:ndhtoda-tau}, the spectral transformations \eqref{eq:st-m1-2} and
\eqref{eq:st-m1-3} for $n=N-1$ read
\begin{equation}\label{eq:phi-N}
  \phi^{(k+M, t)}_N(z)=\phi^{(k, t+1)}_N(z)=\phi^{(k, t)}_N(z).
\end{equation}
The ``dual'' relation
\begin{equation*}
  \psi^{(k+1, t)}_N(z)=\psi^{(k, t)}_N(z)
\end{equation*}
also holds.

By using the $N\times N$ bidiagonal matrices
\begin{gather*}
  \fl
  L^{(k, t)}\coloneq
  \begin{pmatrix}
    1\\
    e^{(k, t)}_0 & 1\\
    & e^{(k, t)}_1 & \ddots\\
    && \ddots & \ddots\\
    &&& e^{(k, t)}_{N-2} & 1
  \end{pmatrix},\quad
  R^{(k, t)}\coloneq
  \begin{pmatrix}
    q^{(k, t)}_0 & 1\\
    & q^{(k, t)}_1 & 1\\
    && \ddots & \ddots\\
    &&& \ddots & 1\\
    &&&& q^{(k, t)}_{N-1}
  \end{pmatrix},\\
  \fl
  \tilde L^{(k, t)}\coloneq
  \begin{pmatrix}
    1\\
    \tilde e^{(k, t)}_0 & 1\\
    & \tilde e^{(k, t)}_1 & \ddots\\
    && \ddots & \ddots\\
    &&& \tilde e^{(k, t)}_{N-2} & 1
  \end{pmatrix},
\end{gather*}
and the $N$-dimensional vectors
\begin{equation*}
  \bm\phi^{(k, t)}(z)\coloneq
  \begin{pmatrix}
    \phi^{(k, t)}_0(z)\\
    \phi^{(k, t)}_1(z)\\
    \vdots\\
    \phi^{(k, t)}_{N-1}(z)
  \end{pmatrix},\quad
  \bm\phi^{(k, t)}_N(z)\coloneq
  \begin{pmatrix}
    0\\
    \vdots\\
    0\\
    \phi^{(k, t)}_{N}(z)
  \end{pmatrix},
\end{equation*}
the spectral transformations~\eqref{eq:st-m1-1}--\eqref{eq:st-m1-3} with
the finite lattice boundary condition~\eqref{eq:finite-lattice-BC} are
written as
\begin{align*}
  z\bm\phi^{(k+1, t)}(z)&=R^{(k, t)}\bm \phi^{(k, t)}(z)+\bm\phi^{(k, t)}_N(z),\\
  \bm\phi^{(k-M, t)}(z)&=L^{(k-M, t)}\bm\phi^{(k, t)}(z),\\
  \bm\phi^{(k, t-1)}(z)&=\tilde L^{(k, t-1)}\bm\phi^{(k, t)}(z),
\end{align*}
respectively.
Hence, we have
\begin{align*}
  z\bm\phi^{(k+1, t)}(z)
  &=\tilde L^{(k+1, t)}R^{(k, t+1)}\bm \phi^{(k, t+1)}(z)+\tilde L^{(k+1, t)}\bm \phi^{(k,  t+1)}_N(z)\\
  &=R^{(k, t)}\tilde L^{(k, t)}\bm\phi^{(k, t+1)}(z)+\bm \phi^{(k, t)}_N(z),\\
  \bm\phi^{(k, t)}(z)
  &=\tilde L^{(k, t)}L^{(k, t+1)}\bm\phi^{(k+M, t+1)}(z)\\
  &=L^{(k, t)}\tilde L^{(k+M, t)}\bm\phi^{(k+M, t+1)}(z).
\end{align*}
Note that, from \eqref{eq:phi-N},
$\tilde L^{(k+1, t)}\bm\phi^{(k, t+1)}_N(z)=\bm\phi^{(k, t+1)}_N(z)=\bm\phi^{(k, t)}_N(z)$
holds.
Therefore, the compatibility conditions are written in the same matrix form
also for the finite lattice case:
\begin{equation}\label{eq:finite-Lax-form}
  \tilde L^{(k+1, t)}R^{(k, t+1)}=R^{(k, t)}\tilde L^{(k, t)},\quad
  \tilde L^{(k, t)}L^{(k, t+1)}=L^{(k, t)}\tilde L^{(k+M, t)}.
\end{equation}

Now consider upper Hessenberg matrices of the form
\begin{equation}\label{eq:upper-Hessenberg}
  H^{(k, t)}\coloneq L^{(k, t)}R^{(k+M-1, t)}R^{(k+M-2, t)}\dots R^{(k, t)}.
\end{equation}
By using the matrix relations~\eqref{eq:finite-Lax-form}, we find
\begin{align*}
  \tilde L^{(k, t)}H^{(k, t+1)}
  &=\tilde L^{(k, t)}L^{(k, t+1)}R^{(k+M-1, t+1)}R^{(k+M-2, t+1)}\dots R^{(k, t+1)}\\
  &=L^{(k, t)}\tilde L^{(k+M, t)}R^{(k+M-1, t+1)}R^{(k+M-2, t+1)}\dots R^{(k, t+1)}\\
  &=L^{(k, t)}R^{(k+M-1, t)}\tilde L^{(k+M-1, t)}R^{(k+M-2, t+1)}\dots R^{(k, t+1)}\\
  &=\dots\\
  &=L^{(k, t)}R^{(k+M-1, t)}R^{(k+M-2, t)}\dots \tilde L^{(k+1, t)}R^{(k, t+1)}\\
  &=L^{(k, t)}R^{(k+M-1, t)}R^{(k+M-2, t)}\dots R^{(k, t)}\tilde L^{(k, t)}\\
  &=H^{(k, t)}\tilde L^{(k, t)}.
\end{align*}
Since $\tilde L^{(k, t)}$ is regular, this implies that
the two upper Hessenberg matrices
$H^{(k, t+1)}$ and $H^{(k, t)}$ are similar.
Therefore, we can say that the ndh-Toda lattice~\eqref{eq:ndhtoda-sf}
with the finite lattice boundary condition~\eqref{eq:finite-lattice-BC}
gives recurrence relations for computing iterations of similarity transformations
of the upper Hessenberg matrices.
Notice that the recurrence relations discussed here are essentially same as
the eigenvalue algorithm for totally nonnegative Hessenberg matrices
proposed by Fukuda \etal~\cite[Algorithm~1]{fukuda2012eam}.
Therefore, we can say that we are investigating
another theoretical aspect of the eigenvalue algorithm
from the viewpoint of discrete integrable systems and
biorthogonal polynomials.

To construct solutions for the system above,
let us consider the theory of ``discrete'' $(M, 1)$-biorthogonal polynomials
$\{\phi^{(k, t)}_n(z)\}_{n=0}^N$ and
$\{\psi^{(k, t)}_n(z)\}_{n=0}^N$ with respect to
$\mathcal L^{(k, t)}$;
i.e. the polynomials satisfy
\begin{gather}
  \begin{multlined}[t][\dimexpr\textwidth-\mathindent-1em]
    \mathcal L^{(k, t)}[z^{mM}\phi^{(k, t)}_n(z)]
    =\mathcal L^{(k, t)}[z^m\psi^{(k, t)}_n(z^M)]
    =\frac{\tau^{(k, t)}_{n+1}}{\tau^{(k, t)}_{n}}\delta_{m, n},\\
    n=0, 1, \dots, N-1,\quad m=0, 1, \dots, n,
  \end{multlined}\nonumber\\
  \mathcal L^{(k, t)}[z^{mM}\phi^{(k, t)}_N(z)]
  =\mathcal L^{(k, t)}[z^m\psi^{(k, t)}_N(z^M)]
  =0,\quad m=0, 1, 2, \dots,\label{eq:d-m1-biorth-rel}
\end{gather}
for all $k, t \in \mathbb Z$.
We can proof the following theorem, which is an analogue of Gauss quadrature
for orthogonal polynomials.

\begin{theorem}\label{th:analog-Gauss}
  Suppose that the monic polynomial
  $\psi^{(k, t)}_N(z)$ has simple zeros
  $z^{(t)}_{0}, z^{(t)}_{1}, \dots, z^{(t)}_{N-1}$.
  Then, there exist constants $w^{(k, t)}_{r, \nu} \in \mathbb C$,
  $r=0, 1, \dots, N-1$ and $\nu=0, 1, \dots, M-1$,
  satisfying
  \begin{equation*}
    \mathcal L^{(k, t)}[\pi(z)]
    =\sum_{r=0}^{N-1}\sum_{\nu=0}^{M-1} w^{(k, t)}_{r, \nu} \pi\left((z^{(t)}_{r})^{1/M}\rme^{-2\pi\rmi\nu/M}\right)
  \end{equation*}
  for all $\pi(z) \in \mathbb C[z]$, where $(z^{(t)}_{r})^{1/M}$ denotes
  one of the $M$th roots of $z^{(t)}_{r}$.
\end{theorem}
\begin{proof}
  Let us consider the Lagrange interpolating polynomial
  \begin{equation*}
    L(z)\coloneq\sum_{r=0}^{N-1}\sum_{\nu=0}^{M-1} \frac{\pi\left((z^{(t)}_{r})^{1/M}\rme^{-2\pi\rmi\nu/M}\right)\psi^{(k, t)}_N(z^M)}{\displaystyle\left.\left(\psi^{(k, t)}_N(z^M)\right)'\right|_{z=(z^{(t)}_{r})^{1/M}\rme^{-2\pi\rmi\nu/M}}\left(z-(z^{(t)}_{r})^{1/M}\rme^{-2\pi\rmi\nu/M}\right)},
  \end{equation*}
  where $'$ indicates the differentiation with respect to $z$.
  Note that we can factorize the polynomial $\psi^{(k, t)}_N(z^M)$ into
  \begin{equation*}
    \psi^{(k, t)}_N(z^M)=\prod_{r=0}^{N-1}\prod_{\nu=0}^{M-1} \left(z-(z^{(t)}_{r})^{1/M}\rme^{-2\pi\rmi\nu/M}\right).
  \end{equation*}
  Hence $L(z)$ is a polynomial at most degree $MN-1$ satisfying
  \begin{multline*}
    L\left((z^{(t)}_{r})^{1/M}\rme^{-2\pi\rmi\nu/M}\right)
    =\pi\left((z^{(t)}_{r})^{1/M}\rme^{-2\pi\rmi\nu/M}\right),\\
    r=0, 1, \dots, N-1,\quad \nu=0, 1, \dots, M-1.
  \end{multline*}
  This implies the existence of a polynomial $P(z)$ (may be zero, if $\deg \pi(z)<MN$)
  satisfying
  \begin{equation*}
    \pi(z)-L(z)=P(z)\psi^{(k, t)}_N(z^M).
  \end{equation*}
  By using the relation above and the biorthogonal relation~\eqref{eq:d-m1-biorth-rel},
  we find
  \begin{gather*}
    \begin{split}
      \mathcal L^{(k, t)}[\pi(z)]
      &=\mathcal L^{(k, t)}[L(z)+P(z)\psi^{(k, t)}_N(z^M)]\\
      &=\mathcal L^{(k, t)}[L(z)]\\
      &=\sum_{r=0}^{N-1}\sum_{\nu=0}^{M-1} w^{(k, t)}_{r, \nu} \pi\left((z^{(t)}_{r})^{1/M}\rme^{-2\pi\rmi\nu/M}\right),
    \end{split}\\
    w^{(k, t)}_{r, \nu}\coloneq\mathcal L^{(k, t)}\left[\frac{\psi^{(k, t)}_N(z^M)}{\displaystyle\left.\left(\psi^{(k, t)}_N(z^M)\right)'\right|_{z=(z^{(t)}_{r})^{1/M}\rme^{-2\pi\rmi\nu/M}}\left(z-(z^{(t)}_{r})^{1/M}\rme^{-2\pi\rmi\nu/M}\right)}\right],
  \end{gather*}
  which completes the proof.
\end{proof}

Theorem~\ref{th:analog-Gauss} leads us to a representation of the moment:
\begin{align}
  \mu^{(t)}_m
  &=\mathcal L^{(0, 0)}\left[z^m\prod_{\tau=0}^{t-1}(z^M-s^{(\tau)})\right]\nonumber\\
  &=\sum_{r=0}^{N-1}\sum_{\nu=0}^{M-1} w^{(0, 0)}_{r, \nu} \left(z_{r}^{1/M}\rme^{-2\pi\rmi\nu/M}\right)^m \prod_{\tau=0}^{t-1}(z_{r}-s^{(\tau)})\nonumber\\
  &=\sum_{r=0}^{N-1}\left(\sum_{\nu=0}^{M-1} w^{(0, 0)}_{r, \nu} \rme^{-2\pi \rmi m\nu/M}\right)z_{r}^{m/M}\prod_{\tau=0}^{t-1}(z_{r}-s^{(\tau)}),\label{eq:ndhtoda-moment-repr-1}
\end{align}
where $z_{r}\coloneq z^{(0)}_{r}$.
Let us introduce new constants
\begin{equation}\label{eq:weight-corr}
  w^{(m)}_{r}\coloneq\sum_{\nu=0}^{M-1} w^{(0, 0)}_{r, \nu} \rme^{-2\pi \rmi m\nu/M},\enspace
  r=0, 1, \dots, N-1,\enspace
  m=0, 1, 2, \dots.
\end{equation}
Then, the representation of the moment~\eqref{eq:ndhtoda-moment-repr-1} is rewritten as
\begin{equation}\label{eq:ndhtoda-moment-repr}
  \mu^{(t)}_m=\sum_{\nu=0}^{N-1}w^{(m)}_{r}z_{r}^{m/M}\prod_{\tau=0}^{t-1}(z_{r}-s^{(\tau)}).
\end{equation}
We should remark that $w_{r}^{(m)}=w_{r}^{(m\bmod M)}$ holds for all $m=0, 1, 2, \dots$ and
there is a one-to-one correspondence between
the constants $\{w^{(0, 0)}_{r, \nu}\}_{r=0, 1, \dots, N-1}^{\nu=0, 1, \dots, M-1}$
and $\{w^{(m)}_{r}\}_{r=0, 1, \dots, N-1}^{m=0, 1, \dots, M-1}$
via the definition~\eqref{eq:weight-corr}, that is the discrete Fourier transform.

Substituting the moment representation~\eqref{eq:ndhtoda-moment-repr}
into the determinant $\tau^{(k, t)}_n$~\eqref{eq:m1-tau-def},
we find
\begin{equation}\label{eq:tau-factorized}
  \tau^{(k, t)}_n=\det(\tilde V^{(k)}_n \mathcal D^{(k, t)} V_n),
\end{equation}
where
\begin{gather*}
  \tilde V^{(k)}_n\coloneq
  \begin{pmatrix}
    w_0^{(k)} & w_1^{(k)} & \dots & w_{N-1}^{(k)}\\
    w_0^{(k+1)} z_0^{1/M} & w_1^{(k+1)}z_1^{1/M} & \dots & w_{N-1}^{(k+1)}z_{N-1}^{1/M}\\
    \vdots & \vdots & & \vdots\\
    w_0^{(k+n-1)} z_0^{(n-1)/M} & w_1^{(k+n-1)}z_1^{(n-1)/M} & \dots & w_{N-1}^{(k+n-1)}z_{N-1}^{(n-1)/M}
  \end{pmatrix},\\
  \mathcal D^{(k, t)}\coloneq\mathop{\mathrm{diag}}\left(z_0^{k/M}\prod_{\tau=0}^{t-1}(z_0-s^{(\tau)}), \dots, z_{N-1}^{k/M}\prod_{\tau=0}^{t-1}(z_{N-1}-s^{(\tau)})\right),\\
  V_n\coloneq
  \begin{pmatrix}
    1 & z_0 & \dots & z_0^{n-1}\\
    1 & z_1 & \dots & z_1^{n-1}\\
    \vdots & \vdots && \vdots\\
    1 & z_{N-1} & \dots & z_{N-1}^{n-1}
  \end{pmatrix}.
\end{gather*}
Applying the Binet--Cauchy formula and the expansion formula for
the Vandermonde determinant to \eqref{eq:tau-factorized}, we obtain
\begin{equation}\label{eq:tau-expanded-v}
  \fl
  \tau^{(k, t)}_n=\sum_{0\le r_0<r_1<\dots<r_{n-1}\le N-1} \mathcal V^{(k)}_{r_0, r_1, \dots, r_{n-1}} \prod_{j=0}^{n-1}\left(z_{r_j}^{k/M}\prod_{\tau=0}^{t-1}(z_{r_j}-s^{(\tau)})\right)\prod_{0\le i<j\le n-1}(z_{r_j}-z_{r_i}),
\end{equation}
where
\begin{equation*}
  \fl
  \mathcal V^{(k)}_{r_0, r_1, \dots, r_{n-1}}\coloneq
  \begin{vmatrix}
    w_{r_0}^{(k)} & w_{r_1}^{(k)} & \dots & w_{r_{n-1}}^{(k)}\\
    w_{r_0}^{(k+1)} z_{r_0}^{1/M} & w_{r_1}^{(k+1)}z_{r_1}^{1/M} & \dots & w_{r_{n-1}}^{(k+1)}z_{r_{n-1}}^{1/M}\\
    \vdots & \vdots & & \vdots\\
    w_{r_0}^{(k+n-1)} z_{r_0}^{(n-1)/M} & w_{r_1}^{(k+n-1)}z_{r_1}^{(n-1)/M} & \dots & w_{r_{n-1}}^{(k+n-1)}z_{r_{n-1}}^{(n-1)/M}
  \end{vmatrix}.
\end{equation*}

Let us derive a sufficient condition for the positivity $\tau^{(k, t)}_n>0$
for all $k$, $t \in \mathbb Z$ and $n=1, 2, \dots, N$.
Hereafter, we suppose that
\begin{itemize}
\item $z_0, z_1, \dots, z_{N-1}$ are all real numbers satisfying $0<z_0<z_1<\dots<z_{N-1}$;
\item all the $M$th roots $z_0^{1/M}, z_1^{1/M}, \dots, z_{N-1}^{1/M}$ are
  chosen as real numbers;
\item the parameter $s^{(t)}$ is chosen as $s^{(t)}<z_0$ for all $t \in \mathbb Z$.
\end{itemize}
In addition, if $\mathcal V^{(k)}_{r_0, r_1, \dots, r_{n-1}}>0$ for
all $k \in \mathbb Z$ and
all $n$-tuples $(r_0, r_1, \dots, r_{n-1})$ satisfying $0\le r_0<r_1<\dots<r_{n-1}\le N-1$,
$n=1, 2, \dots, N$, then it is obvious that the conditions
are sufficient for the positivity of $\tau^{(k, t)}_n$.
Since $w_r^{(m)}=w_r^{(m \bmod M)}$
implies
$\mathcal V^{(k)}_{r_0, r_1, \dots, r_{n-1}}=\mathcal V^{(k \bmod M)}_{r_0, r_1, \dots, r_{n-1}}$,
the number of the conditions $\mathcal V^{(k)}_{r_0, r_1, \dots, r_{n-1}}>0$
is finite: there are $M\sum_{n=1}^{N} \binom{N}{n}=M(2^N-1)$ conditions.

We will rewrite the condition $\mathcal V^{(k)}_{r_0, r_1, \dots, r_{n-1}}>0$
in a simpler form.
First, if $n=1$, then $\mathcal V^{(k)}_{r_0}=w_{r_0}^{(k)}>0$;
i.e. all $w_{r}^{(m)}$ must be positive.
Next, if $n=2, 3, \dots, N$, then the elementary row-additions yield
\begin{flalign*}
  \mathcal V^{(k)}_{r_0, r_1, \dots, r_{n-1}}&=
  \begin{vmatrix}
    w_{r_0}^{(k)} & w_{r_1}^{(k)} & w_{r_2}^{(k)} & \dots & w_{r_{n-1}}^{(k)}\\
    0 & w_{r_0, r_1}^{(k+1)} & w_{r_0, r_2}^{(k+1)}& \dots & w_{r_0, r_{n-1}}^{(k+1)}\\
    0 & w_{r_0, r_1}^{(k+2)}z_{r_1}^{1/M} & w_{r_0, r_2}^{(k+2)}z_{r_2}^{1/M} & \dots & w_{r_0, r_{n-1}}^{(k+1)}z_{r_{n-1}}^{1/M}\\
    \vdots & \vdots & \vdots & & \vdots\\
    0 & w_{r_0, r_1}^{(k+n-1)}z_{r_1}^{(n-2)/M} & w_{r_0, r_2}^{(k+n-1)}z_{r_2}^{(n-2)/M} & \dots & w_{r_0, r_{n-1}}^{(k+n-1)}z_{r_{n-1}}^{(n-2)/M}
  \end{vmatrix}&\\
  &=w_{r_0}^{(k)}
  \begin{vmatrix}
    w_{r_0, r_1}^{(k+1)} & w_{r_0, r_2}^{(k+1)}& \dots & w_{r_0, r_{n-1}}^{(k+1)}\\
    w_{r_0, r_1}^{(k+2)}z_{r_1}^{1/M} & w_{r_0, r_2}^{(k+2)}z_{r_2}^{1/M} & \dots & w_{r_0, r_{n-1}}^{(k+2)}z_{r_{n-1}}^{1/M}\\
    \vdots & \vdots & & \vdots\\
    w_{r_0, r_1}^{(k+n-1)}z_{r_1}^{(n-2)/M} & w_{r_0, r_2}^{(k+n-1)}z_{r_2}^{(n-2)/M} & \dots & w_{r_0, r_{n-1}}^{(k+n-1)}z_{r_{n-1}}^{(n-2)/M}
  \end{vmatrix},&
\end{flalign*}
where
\begin{equation*}
  \fl
  w_{r_0, r_1}^{(m+1)}
  \coloneq\frac{w_{r_0}^{(m)}w_{r_1}^{(m+1)}z_{r_1}^{1/M}-w_{r_0}^{(m+1)}w_{r_1}^{(m)}z_{r_0}^{1/M}}{w_{r_0}^{(m)}}
  =\frac{1}{w_{r_0}^{(m)}}
  \begin{vmatrix}
    w_{r_0}^{(m)} & w_{r_1}^{(m)}\\
    w_{r_0}^{(m+1)}z_{r_0}^{1/M} & w_{r_1}^{(m+1)}z_{r_1}^{1/M}
  \end{vmatrix}.
\end{equation*}
In the same manner, we can show by induction on $n$ that
\begin{equation}\label{eq:vtow}
  \mathcal V_{r_0, r_1, \dots, r_{n-1}}^{(k)}=w_{r_0}^{(k)}w_{r_0, r_1}^{(k+1)}\dots w_{r_0, r_1, \dots, r_{n-1}}^{(k+n-1)},
\end{equation}
where $w_{r_0, r_1, \dots, r_{n-1}}^{(m)}$ is defined recursively by
\begin{equation}\label{eq:w-def}
  \fl
  w_{r_0, \dots, r_{n-3}, r_{n-2}, r_{n-1}}^{(m+1)}
  \coloneq\frac{1}{w_{r_0, \dots, r_{n-3}, r_{n-2}}^{(m)}}
  \begin{vmatrix}
    w_{r_0, \dots, r_{n-3}, r_{n-2}}^{(m)} & w_{r_0, \dots, r_{n-3}, r_{n-1}}^{(m)}\\
    w_{r_0, \dots, r_{n-3}, r_{n-2}}^{(m+1)}z_{r_{n-2}}^{1/M} & w_{r_0, \dots, r_{n-3}, r_{n-1}}^{(m+1)}z_{r_{n-1}}^{1/M}
  \end{vmatrix}.
\end{equation}
From equation~\eqref{eq:vtow} with the condition
$\mathcal V_{r_0, r_1, \dots, r_{n-1}}^{(k)}>0$, it is readily induced by induction on $n$
that all $w_{r_0, r_1, \dots, r_{n-1}}^{(m)}$ must be positive.
Thus we obtain the following theorem.

\begin{theorem}
  Suppose that all $w_r^{(m)}$ are positive and the relation
  \begin{equation}\label{eq:cond-positivity}
    \fl
    w_{r_0, \dots, r_{n-3}, r_{n-2}}^{(m)}w_{r_0, \dots, r_{n-3}, r_{n-1}}^{(m+1)}z_{r_{n-1}}^{1/M}
    >w_{r_0, \dots, r_{n-3}, r_{n-2}}^{(m+1)}w_{r_0, \dots, r_{n-3}, r_{n-1}}^{(m)}z_{r_{n-2}}^{1/M}
  \end{equation}
  holds for all $m=0, 1, \dots, M-1$ and all $n$-tuples $(r_0, r_1, \dots, r_{n-1})$ satisfying $0\le r_0<r_1<\dots<r_{n-1}\le N-1$, $n=2, 3, \dots, N$.
  Then,
  \begin{multline}\label{eq:tau-expanded}
    \fl
    \tau^{(k, t)}_n=\sum_{0\le r_0<r_1<\dots<r_{n-1}\le N-1} \Bigg(w_{r_0}^{(k)}w_{r_0, r_1}^{(k+1)}\dots w_{r_0, r_1, \dots, r_{n-1}}^{(k+n-1)} \\
    \times\prod_{j=0}^{n-1}\left(z_{r_j}^{k/M}\prod_{\tau=0}^{t-1}(z_{r_j}-s^{(\tau)})\right)\prod_{0\le i<j\le n-1}(z_{r_j}-z_{r_i})\Bigg)>0
  \end{multline}
  for all $k$, $t \in \mathbb Z$ and $n=1, 2, \dots, N$.
\end{theorem}

\begin{corollary}
  A solution to the ndh-Toda lattice~\eqref{eq:ndhtoda-sf}
  with the finite lattice boundary condition~\eqref{eq:ndhtoda-finite-cond}
  is given by \eqref{eq:ndhtoda-tau} and \eqref{eq:tau-expanded}.
  If the positivity condition~\eqref{eq:cond-positivity} is satisfied,
  then the variables
  $q^{(k, t)}_n$, $e^{(k, t)}_n$, $\tilde e^{(k, t)}_n$ and $d^{(k, t)}_n$
  are always positive. Furthermore, if the parameter $s^{(t)}$ is chosen as
  $s^{(t)}\le 0$, then $f^{(k, t)}_n$ is always nonnegative.
\end{corollary}

Finally, we discuss the asymptotic behaviour of the solution as $t \to \infty$.
Since, from the assumption, $0<z_0-s^{(t)}<z_1-s^{(t)}<\dots<z_{N-1}-s^{(t)}$ holds
for all $t \in \mathbb Z$,
we find
\begin{equation*}
  \fl
  \tau^{(k, t)}_n \sim \mathcal V^{(k)}_{N-n, N-n+1, \dots, N-1}\prod_{j=0}^{n-1}\left(z_{N-n+j}^{k/M} \prod_{\tau=0}^{t-1}(z_{N-n+j}-s^{(\tau)})\right)\prod_{0\le i<j\le n-1}(z_{N-n+j}-z_{N-n+i})
\end{equation*}
as $t \to \infty$.
Hence, we have the following asymptotic behaviour of
the solution given by \eqref{eq:ndhtoda-tau} and \eqref{eq:tau-expanded-v}:
\begin{subequations}
  \begin{gather}
    q^{(k, t)}_n \to \frac{\mathcal V^{(k)}_{N-n, N-n+1, \dots, N-1}\mathcal V^{(k+1)}_{N-n-1, N-n, \dots, N-1}}{\mathcal V^{(k)}_{N-n-1, N-n, \dots, N-1}\mathcal V^{(k+1)}_{N-n, N-n+1, \dots, N-1}}z_{N-n-1}^{1/M},\\
    e^{(k, t)}_n \sim \gamma^{(k)}_n \prod_{\tau=0}^{t-1}\frac{z_{N-n-2}-s^{(\tau)}}{z_{N-n-1}-s^{(\tau)}}\to 0\label{eq:asymptotics-e}
  \end{gather}
\end{subequations}
as $t \to \infty$, where
\begin{equation*}
  \fl
  \gamma^{(k)}_n\coloneq \frac{\mathcal V^{(k)}_{N-n-2, N-n-1, \dots, N-1}\mathcal V^{(k)}_{N-n, N-n+1, \dots, N-1}}{(\mathcal V^{(k)}_{N-n-1, N-n, \dots, N-1})^2}\cdot\frac{z_{N-n-2}^{k/M}}{z_{N-n-1}^{(k+M)/M}}\cdot\frac{\prod_{j=0}^{n}(z_{N-n-1+j}-z_{N-n-2})}{\prod_{j=0}^{n-1}(z_{N-n+j}-z_{N-n-1})}.
\end{equation*}
Especially, we also have
\begin{equation*}
  \prod_{j=0}^{M-1} q^{(k+j, t)}_n \to z_{N-n-1}
\end{equation*}
as $t \to \infty$.

The results indicate that the upper Hessenberg matrix
$H^{(k, t)}$~\eqref{eq:upper-Hessenberg} goes
to an upper triangular matrix whose $(n, n)$-entry is $z_{N-n-1}$ as
$t \to \infty$, $n=0, 1, \dots, N-1$. Since $H^{(k, t)}$ and $H^{(k, t+1)}$ are similar,
it is revealed that $z_0, z_1, \dots, z_{N-1}$ are the eigenvalues of
$H^{(k, t)}$. That is, the recurrence relations of the ndh-Toda
lattice~\eqref{eq:dhtoda} with the finite lattice
boundary condition~\eqref{eq:ndhtoda-finite-cond} give an eigenvalue
algorithm for upper Hessenberg matrices that can be factorized into
a product of bidiagonal matrices as \eqref{eq:upper-Hessenberg}.
Its convergence speed depends on the value $(z_{N-n-2}-s^{(t)})/(z_{N-n-1}-s^{(t)})$,
$n=0, 1, \dots, N-2$, by \eqref{eq:asymptotics-e}.
This means that an appropriate choice of the parameter $s^{(t)}$
may improve the convergence speed.

\section{Ultradiscretization}\label{sec:ultradiscretization}
In this section, we ultradiscretize the ndh-Toda lattice and its solution
in \sref{sec:m-1-reduction},
and give a proof of a connection between the derived ultradiscrete system and
the generalized BBS.

\subsection{Nonautonomous ultradiscrete finite hungry Toda lattice}
For the variables and parameter of the ndh-Toda lattice~\eqref{eq:ndhtoda-sf},
we consider the transformations from the variables to new variables
denoted by capital letters as follows:
$q^{(k, t)}_n=\rme^{-Q^{(k, t)}_n/\epsilon}$,
$e^{(k, t)}_n=\rme^{-E^{(k, t)}_n/\epsilon}$,
$\tilde e^{(k, t)}_n=\rme^{-\tilde E^{(k, t)}_n/\epsilon}$,
$d^{(k, t)}_n=\rme^{-D^{(k, t)}_n/\epsilon}$,
$f^{(k, t)}_n=\rme^{-F^{(k, t)}_n/\epsilon}$ and
$s^{(t)}=-\rme^{-S^{(t)}/\epsilon}$, where $\epsilon$ is
a positive parameter.
Since there is an ultradiscretization formula
\begin{equation*}
  \lim_{\epsilon \to +0} -\epsilon\log(p_1\rme^{-A/\epsilon}+p_2\rme^{-B/\epsilon})=\min(A, B),
\end{equation*}
where $p_1$ and $p_2$ are positive numbers,
applying these transformations and taking a limit $\epsilon \to +0$ yield
piecewise linear recurrence relations
\begin{subequations}\label{eq:nuhtoda}
  \begin{gather}
    Q^{(k, t+1)}_n=\min(D^{(k, t)}_n, \tilde E^{(k, t)}_n),\label{eq:nuhtoda-q}\\
    E^{(k, t+1)}_n=\min(F^{(k, t)}_n, \tilde E^{(k+M, t)}_n),\label{eq:nuhtoda-e}\\
    D^{(k, t)}_{n+1}=D^{(k, t)}_n-Q^{(k, t+1)}_n+Q^{(k, t)}_{n+1},\label{eq:nuhtoda-d}\\
    F^{(k, t)}_{n+1}=F^{(k, t)}_n-E^{(k, t+1)}_n+E^{(k, t)}_{n+1},\label{eq:nuhtoda-f}\\
    \tilde E^{(k+1, t)}_n=\tilde E^{(k, t)}_n-Q^{(k, t+1)}_n+Q^{(k, t)}_{n+1},\label{eq:nuhtoda-te1}\\
    \tilde E^{(k, t)}_{n+1}=\tilde E^{(k+M, t)}_n-E^{(k, t+1)}_n+E^{(k, t)}_{n+1}\label{eq:nuhtoda-te2}
  \end{gather}
  for $n=0, 1, 2, \dots$ with the boundary condition
  \begin{gather}
    D^{(k, t)}_0=Q^{(k, t)}_0,\label{eq:nuhtoda-bc-d}\\
    F^{(k, t)}_0=E^{(k, t)}_0+\max\left(0, S^{(t)}-\sum_{j=0}^{M-1} Q^{(k+j, t)}_0\right),\label{eq:nuhtoda-bc-f}\\
    \tilde E^{(k, t)}_0=E^{(k, t)}_0+\max\left(0, \sum_{j=0}^{M-1} Q^{(k+j, t)}_0-S^{(t)}\right)\label{eq:nuhtoda-bc-te}
  \end{gather}
  for all $k, t \in \mathbb Z$.
\end{subequations}
We call the system~\eqref{eq:nuhtoda}
the \emph{nonautonomous ultradiscrete hungry Toda lattice} (nuh-Toda lattice).
In addition, we also impose the finite lattice condition corresponding
to \eqref{eq:ndhtoda-finite-cond}:
\begin{equation}\label{eq:nuhtoda-finite-cond}
  E^{(k, t)}_{N-1}=\tilde E^{(k, t)}_{N-1}=+\infty.
\end{equation}

A solution to the nuh-Toda lattice~\eqref{eq:nuhtoda} with
the finite lattice condition~\eqref{eq:nuhtoda-finite-cond} is
constructed from the solution~\eqref{eq:ndhtoda-tau} and \eqref{eq:tau-expanded}
to the ndh-Toda lattice~\eqref{eq:ndhtoda-sf} with \eqref{eq:ndhtoda-finite-cond}.
Consider the transformations of variables
$\tau^{(k, t)}_n=\rme^{-T^{(k, t)}_n/\epsilon}$,
$z_n=p_n\rme^{-Z_n/\epsilon}$ and
$w^{(m)}_n=\rme^{-W^{(m)}_n/\epsilon}$
and the limit procedure $\epsilon \to +0$,
where $p_n$ is a positive constant satisfying
$p_n<p_{n+1}$ if $Z_n=Z_{n+1}$.
Note that, since we assume the inequality $0<z_0<z_1<\dots<z_{N-1}$ in
\sref{sec:m-1-reduction}, the new variable $Z_n$ satisfies
$Z_0\ge Z_1 \ge \dots \ge Z_{N-1}$.
To apply the transformations of variables, $\tau^{(k, t)}_n$ must be positive;
i.e. the condition~\eqref{eq:cond-positivity} must be satisfied.
This means that the new variables satisfy the relation
\begin{equation}\label{eq:u-cond-positivity}
  \fl
  W_{r_0, \dots, r_{n-3}, r_{n-2}}^{(m)}+W_{r_0, \dots, r_{n-3}, r_{n-1}}^{(m+1)}+\frac{1}{M}Z_{r_{n-1}}\le W_{r_0, \dots, r_{n-3}, r_{n-2}}^{(m+1)}+W_{r_0, \dots, r_{n-3}, r_{n-1}}^{(m)}+\frac{1}{M}Z_{r_{n-2}}
\end{equation}
for all $m=0, 1, \dots, M-1$ and all $n$-tuples $(r_0, r_1, \dots, r_{n-1})$
satisfying $0\le r_0<r_1<\dots<r_{n-1}\le N-1$, $n=1, 2, 3, \dots, N$.
We should remark that the formula
\begin{equation*}
  \fl
  \lim_{\epsilon\to+0}-\epsilon\log(p_1\rme^{-A/\epsilon}-p_2\rme^{-B/\epsilon})=
  \begin{cases*}
    A & if $A<B$ or $A=B$ and $p_1>p_2$,\\
    \text{indefinite} & if $A>B$ or $A=B$ and $p_1\le p_2$,
  \end{cases*}
\end{equation*}
holds,
where $p_1$ and $p_2$ are positive numbers.
The latter indefinite result is the cause of the
\emph{negative problem of ultradiscretization}.
However, if the condition~\eqref{eq:u-cond-positivity} is satisfied,
then it is assured that we can always use the former result for \eqref{eq:w-def}.
Hence we obtain, by induction on $n$,
\begin{equation}\label{eq:ud-W}
  \fl
  W_{r_0, \dots, r_{n-3}, r_{n-2}, r_{n-1}}^{(m+1)}
  =W_{r_0, \dots, r_{n-3}, r_{n-1}}^{(m+1)}+\frac{1}{M}Z_{r_{n-1}}
  =\dots
  =W_{r_{n-1}}^{(m+1)}+\frac{n-1}{M}Z_{r_{n-1}}.
\end{equation}
We should remark that, by using \eqref{eq:ud-W},
the condition~\eqref{eq:u-cond-positivity} is simply rewritten as
\begin{equation*}
  W^{(m)}_{r_{n-2}}+W^{(m+1)}_{r_{n-1}}+\frac{1}{M}Z_{r_{n-1}}
  \le W^{(m+1)}_{r_{n-2}}+W^{(m)}_{r_{n-1}}+\frac{1}{M}Z_{r_{n-2}}.
\end{equation*}
Hence, we obtain the following theorem.

\begin{theorem}\label{th:nuhtoda-sol}
If the conditions $Z_0\ge Z_1 \ge \dots \ge Z_{N-1}$ and
\begin{equation*}
  W^{(m)}_{r_1}-W^{(m+1)}_{r_1}+W^{(m+1)}_{r_0}-W^{(m)}_{r_0}+\frac{Z_{r_0}-Z_{r_1}}{M}\ge 0
\end{equation*}
are satisfied for all $m=0, 1, \dots, M-1$ and all pairs $(r_0, r_1)$
satisfying $0\le r_0<r_1\le N-1$, then we have, from~\eqref{eq:tau-expanded}
and \eqref{eq:ud-W},
\begin{gather*}
  \fl
  T^{(k, t)}_0=0,\\
  \fl
  \begin{multlined}
    T^{(k, t)}_n=\min_{0\le r_0<r_1<\dots<r_{n-1}\le N-1}\left(\sum_{j=0}^{n-1}\left(W_{r_j}^{(k+j)}+\frac{k+(M+1)j}{M}Z_{r_{j}}+\sum_{\tau=0}^{t-1}\min(Z_{r_j}, S^{(\tau)})\right)\right),\\
    n=1, 2, \dots, N.
  \end{multlined}
\end{gather*}
By using the function $T^{(k, t)}_n$, a solution to the nuh-Toda lattice~\eqref{eq:nuhtoda}
with the finite lattice condition~\eqref{eq:nuhtoda-finite-cond} is given by,
from~\eqref{eq:ndhtoda-tau},
\begin{gather*}
  Q^{(k, t)}_n=T^{(k, t)}_n-T^{(k, t)}_{n+1}+T^{(k+1, t)}_{n+1}-T^{(k+1, t)}_n,\\
  E^{(k, t)}_n=T^{(k, t)}_{n+2}-T^{(k, t)}_{n+1}+T^{(k+M, t)}_{n}-T^{(k+M, t)}_{n+1},\\
  \tilde E^{(k, t)}_n=T^{(k, t)}_{n+2}-T^{(k, t)}_{n+1}+T^{(k, t+1)}_{n}-T^{(k, t+1)}_{n+1},\\
  D^{(k, t)}_n=T^{(k, t+1)}_n-T^{(k, t)}_{n+1}+T^{(k+1, t)}_{n+1}-T^{(k+1, t+1)}_n,\\
  F^{(k, t)}_n=T^{(k, t)}_{n+2}-T^{(k, t+1)}_{n+1}+T^{(k+M, t+1)}_{n}-T^{(k+M, t)}_{n+1}+S^{(t)}.
\end{gather*}
\end{theorem}

\subsection{Connection to the generalized BBS}

Finally, we prove a correspondence between the generalized BBS and
the nuh-Toda lattice~\eqref{eq:nuhtoda}.

The time evolution equation of the BBS with $M$ kinds of balls
and the carrier of capacity $S^{(t)}>0$ at time $t$ is given by
an $(M+1)$-reduced nonautonomous ultradiscrete KP lattice:
\begin{subequations}\label{eq:u-M1-reduced-KP}
  \begin{gather}
    U^{(k, t+1)}_n=U^{(k, t)}_n-X^{(k, t)}_n+X^{(k+1, t)}_n,\\
    V^{(k, t)}_{n+1}=V^{(k, t)}_n+X^{(k, t)}_n-X^{(k+1, t)}_n,\\
    X^{(k, t)}_n=\min_{i=0,1, \dots, M}\left(\sum_{j=0}^{M-i-1}U^{(k+j, t)}_n+\sum_{j=M-i+1}^M V^{(k+j, t)}_n\right),
  \end{gather}
  where $U^{(k+M+1, t)}_n=U^{(k, t)}_n$ and $V^{(k+M+1, t)}_n=V^{(k, t)}_n$ for all
  $k, t, n \in \mathbb Z$, with the boundary condition
  \begin{alignat}{3}
    &U^{(0, t)}_n=1,&\quad&V^{(0, t)}_n=S^{(t)},&\quad&\\
    &U^{(k, t)}_n=0,&&V^{(k, t)}_n=0,&&k=1, 2, \dots, M,
  \end{alignat}
  for $|n|\gg 1$.
\end{subequations}
We should choose the initial values of the system~\eqref{eq:u-M1-reduced-KP} to satisfy
\begin{equation*}
  \sum_{j=0}^{M} U^{(j, 0)}_n=1
\end{equation*}
for all $n \in \mathbb Z$. Then, it is readily shown that the relations
\begin{equation*}
  \sum_{j=0}^{M} U^{(j, t)}_n=1,\quad
  \sum_{j=0}^{M} V^{(j, t)}_n=S^{(t)}
\end{equation*}
hold for all $n, t \in \mathbb Z$.
The variables denote
\begin{itemize}
\item $U^{(0, t)}_n\in\{0, 1\}$: the number of empty spaces in the $n$th box at time $t$;
\item $U^{(k, t)}_n\in\{0, 1\}$: the number of balls with index $k$ in the $n$th box at time $t$,
  $k=1, 2, \dots, M$;
\item $V^{(0, t)}_n\in\{0, 1, \dots, S^{(t)}\}$: the number of empty spaces in the carrier at the $n$th box
  from time $t$ to $t+1$;
\item $V^{(k, t)}_n\in\{0, 1, \dots, S^{(t)}\}$: the number of balls with index $k$ in the carrier at the $n$th box from time $t$ to $t+1$,
  $k=1, 2, \dots, M$.
\end{itemize}

\begin{figure}
  \centering
  \includegraphics[width=\textwidth]{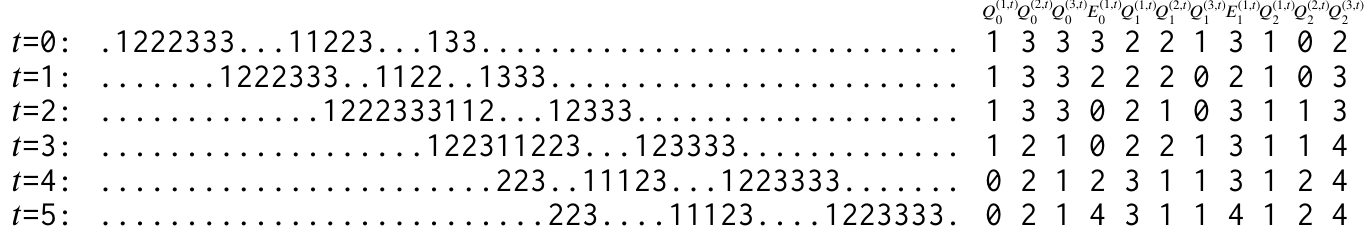}
  \caption{An example of the time evolution of (left) the BBS~\eqref{eq:u-M1-reduced-KP}
    with $M=3$ kinds of balls and carrier capacity $S^{(t)}=6$ for all $t$,
    and of (right) the nuh-Toda lattice~\eqref{eq:nuhtoda}
    with the finite lattice boundary condition $E^{(1, t)}_2=+\infty$.
    As stated in Theorem~\ref{th:corr-bbs-hungry},
    we can see a correspondence between the states of the BBS
    and of the nuh-Toda lattice.}
  \label{fig:ex-bbs}
\end{figure}

The left side of \fref{fig:ex-bbs} shows an example of
the time evolution of the BBS~\eqref{eq:u-M1-reduced-KP}
with $M=3$ kinds of balls and carrier capacity $S^{(t)}=6$ for all $t$,
in which `\texttt{1}', `\texttt{2}', `\texttt{3}', and `\texttt{.}' denote
a ball with index 1, 2, 3, and an empty box, respectively.
Each box can contain only one ball.
In the followings, we regard ``an empty space'' as a ball with index 0.
Then, we can explain the evolution rule of the BBS~\eqref{eq:u-M1-reduced-KP} as follows:
From time $t$ to $t+1$, the carrier of capacity $S^{(t)}$
moves from left to right.
When the carrier passes each box, if the box contains a ball with
index $k$, then the carrier exchanges the ball with
a ball in the carrier whose index is the smallest in the carrier's balls,
where the index order is defined on $\mathbb Z/(M+1)\mathbb Z$ as
the smallest index is $k+1$;
e.g. if $k=2$, then the index order is $3<4<5<\dots<M<0<1<2$.
\begin{figure}
  \centering
  \includegraphics[width=\textwidth]{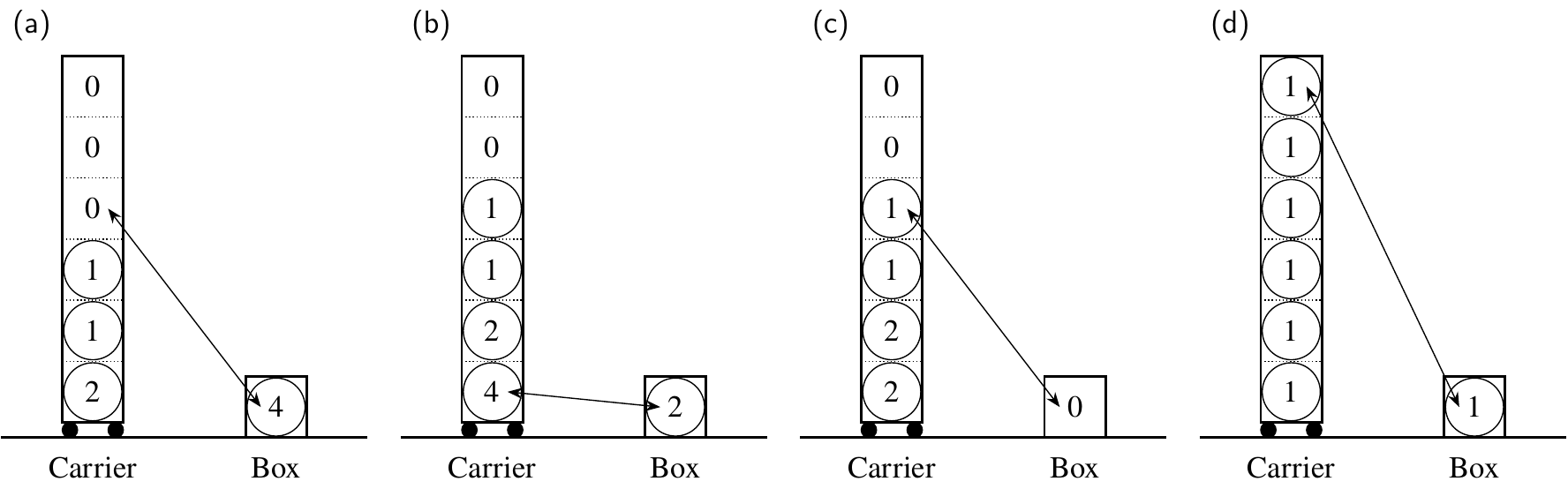}
  \caption{Examples of the exchange rule of balls.}
  \label{fig:exchange-rule}
\end{figure}
\Fref{fig:exchange-rule} also illustrates the exchange rule by examples.

Now, we can prove the following theorem.
\begin{theorem}\label{th:corr-bbs-hungry}
  For the BBS with $M$ kinds of balls and carrier capacity
  $S^{(t)}$ at each time $t$, let
  \begin{itemize}
  \item $Q^{(k, t)}_n$ be the number of balls with index $k$ in the $n$th block of balls
    at time $t$, $k=1, 2, \dots, M$;
  \item $E^{(1, t)}_n$ be the number of empty boxes between the $n$th and $(n+1)$st
    blocks of balls at time $t$.
  \end{itemize}
  Then, the variables $Q^{(k, t)}_n$ and $E^{(1, t)}_n$ satisfy
  the nuh-Toda lattice~\eqref{eq:nuhtoda} with
  the finite lattice condition \eqref{eq:nuhtoda-finite-cond},
  in which $N$ denotes the number of the blocks of balls.
\end{theorem}

The right side of \fref{fig:ex-bbs} shows
an example of the time evolution of the nuh-Toda lattice
with $M=3$, $N=3$ and $S^{(t)}=6$ for all $t$.
The initial values of the nuh-Toda lattice are chosen to correspond
to the initial state of the BBS.
The solution given by Theorem~\ref{th:nuhtoda-sol} with the parameters
$Z_0=7$, $Z_1=5$, $Z_2=3$, $W^{(0)}_0=1$, $W^{(1)}_0=5/3$,
$W^{(2)}_0=1/3$, $W^{(0)}_1=6$, $W^{(1)}_1=19/3$, $W^{(2)}_1=17/3$, $W^{(0)}_2=13$,
$W^{(1)}_2=13$, and $W^{(2)}_2=12$
corresponds to the time evolution in \fref{fig:ex-bbs}.
Notice that the balls in each block must be arranged in ascending order
of indices from left to right.
For example, `\texttt{1222333112}' is composed of two blocks `\texttt{1222333}' and
`\texttt{112}'.

\begin{proof}[Proof of Theorem~\ref{th:corr-bbs-hungry}]
  We will also show the roles of the other variables.
  \begin{itemize}
  \item $D^{(k, t)}_n$: the maximum number of balls with index $k$ that
    the carrier can put into boxes as the part of the $n$th block of balls at time $t+1$,
    $k=1, 2, \dots, M$;
  \item $\tilde E^{(k, t)}_n$: the number of boxes between the leftmost position of
    the balls with index $k$ corresponding to
    the variables $Q^{(k, t+1)}_n$ and $Q^{(k, t)}_{n+1}$, $k=1, 2, \dots, M$;
  \item $\tilde E^{(1+M, t)}_n$: the number of boxes between the leftmost position of
    the empty boxes corresponding to
    the variables $E^{(1, t+1)}_n$ and $E^{(1, t)}_{n+1}$;
  \item $F^{(1, t)}_n$: the sum of the value $E^{(1, t)}_n$ and
    the number of empty spaces in the carrier after passing the $n$th block of
    balls from time $t$ to $t+1$.
  \end{itemize}
  \begin{figure}
    \centering
    \includegraphics{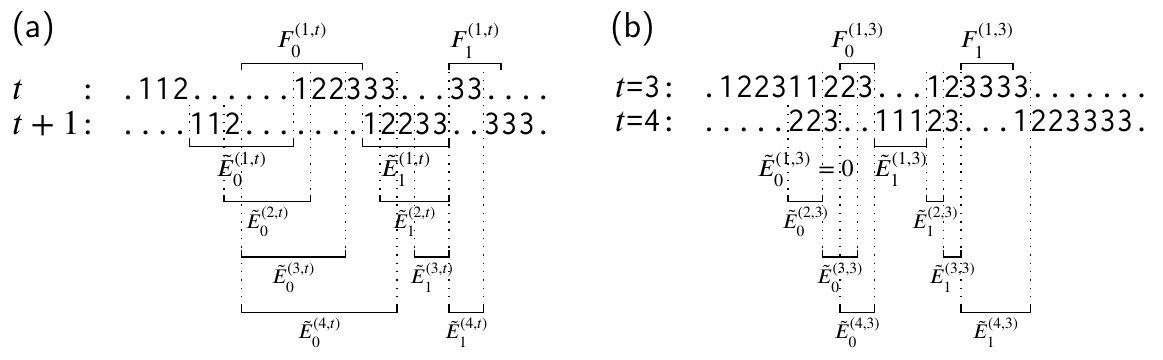}
    \caption{Illustration of the role of the variables $\tilde E^{(k, t)}_n$ and $F^{(1, t)}_n$.
    (a) An example of the case $M=3$, $N=3$ and $S^{(t)}=4$.
    (b) Another example in which the states are excerpted from \fref{fig:ex-bbs},
      $M=3$, $N=3$ and $S^{(3)}=6$.}
    \label{fig:illust-role}
  \end{figure}
  See also~\fref{fig:illust-role}.
  In this proof, we will use the following simple formulae:
  \begin{gather*}
    -\min(-A, -B)=\max(A, B),\\
    A+\min(B, C)=\min(A+B, A+C).
  \end{gather*}

  First, since the carrier gets $Q^{(k, t)}_0$ balls with index $k$
  from the $0$th block of balls, the boundary condition~\eqref{eq:nuhtoda-bc-d}
  gives the number of balls with index $k$ that the carrier can put
  into boxes as the part of the $0$th block of balls at time $t+1$.
  The carrier puts $Q^{(k, t+1)}_n$ balls with index $k$ and
  gets $Q^{(k, t)}_{n+1}$ balls with index $k$ between
  the leftmost position of the $n$th and ($n+1$)st blocks of empty boxes.
  Hence, the recurrence relation~\eqref{eq:nuhtoda-d}
  indeed calculates the value of $D^{(k, t)}_n$ for its role correctly.

  Next, since the carrier exchanges $\min(\sum_{k=1}^{M} Q^{(k, t)}_0, S^{(t)})$
  balls in the 0th block of balls with empty spaces in the carrier,
  the value of $\tilde E^{(1, t)}_0$ is given by
  \begin{align*}
    \tilde E^{(1, t)}_0
    &=E^{(1, t)}_0+\sum_{k=1}^{M} Q^{(k, t)}_0-\min\left(\sum_{k=1}^{M} Q^{(k, t)}_0, S^{(t)}\right)\\
    &=E^{(1, t)}_0+\max\left(0, \sum_{k=1}^{M} Q^{(k, t)}_0-S^{(t)}\right).
  \end{align*}
  This is the boundary condition~\eqref{eq:nuhtoda-bc-te}.
  Then, it is obvious that
  the recurrence relations~\eqref{eq:nuhtoda-te1} and \eqref{eq:nuhtoda-te2}
  indeed calculate the value of $\tilde E^{(k, t)}_n$ for its role.
  Similarly, the value of $F^{(1, t)}_0$ should be
  \begin{align*}
    F^{(1, t)}_0
    &=E^{(1, t)}_0+S^{(t)}-\min\left(\sum_{k=1}^{M} Q^{(k, t)}_0, S^{(t)}\right)\\
    &=E^{(1, t)}_0+\max\left(0, S^{(t)}-\sum_{k=1}^{M} Q^{(k, t)}_0\right),
  \end{align*}
  which coincides with the boundary condition~\eqref{eq:nuhtoda-bc-f}.
  Here, let $Y^{(1, t)}_n$ be
  the number of empty spaces in the carrier just after passing the $n$th block of
  balls at time $t$; i.e.
  \begin{equation*}
    F^{(1, t)}_n=E^{(1, t)}_n+Y^{(1, t)}_n.
  \end{equation*}
  Then, since the carrier puts $\min(E^{(1, t)}_n, S^{(t)}-Y^{(1, t)}_n)$ balls
  into the $n$th block of empty boxes at time $t$,
  the number of empty spaces in the carrier
  just before the carrier passes the ($n+1$)st block of balls is given
  by $Y^{(1, t)}_n+\min(E^{(1, t)}_n, S^{(t)}-Y^{(1, t)}_n)=\min(F^{(1, t)}_n, S^{(t)})$.
  The carrier passes
  $E^{(1, t)}_n-\min(E^{(1, t)}_n, S^{(t)}-Y^{(1, t)}_n)=\max(0, F^{(1, t)}_n-S^{(t)})$
  empty boxes with no balls between the $n$th and ($n+1$)st blocks of balls.
  Hence, the carrier exchanges $E^{(1, t+1)}_n-\max(0, F^{(1, t)}_n-S^{(t)})$ balls
  in the ($n+1$)st block of balls with empty spaces in the carrier.
  We obtain
  \begin{align*}
    Y^{(1, t)}_{n+1}
    &=\min(F^{(1, t)}_n, S^{(t)})-\left(E^{(1, t+1)}_n-\max(0, F^{(1, t)}_n-S^{(t)})\right)\\
    &=F^{(1, t)}_n-E^{(1, t+1)}_n
  \end{align*}
  and
  \begin{align*}
    F^{(1, t)}_{n+1}
    &=E^{(1, t)}_{n+1}+Y^{(1, t)}_{n+1}\\
    &=F^{(1, t)}_n-E^{(1, t+1)}_n+E^{(1, t)}_{n+1},
  \end{align*}
  which coincides with the recurrence relation~\eqref{eq:nuhtoda-f}.

  Notice that the value of $\tilde E^{(k, t)}_n$ gives
  the maximum number of boxes in which the carrier can put the balls
  with index $k$ as the part of the $n$th block of balls at time $t+1$.
  Hence, the value of $Q^{(k, t+1)}_n$ is given by
  the recurrence relation~\eqref{eq:nuhtoda-q}.
  We also notice that $\tilde E^{(1+M, t)}_n$ gives
  the maximum number of boxes which can be consist of $n$th block of empty boxes
  at time $t+1$; i.e. $E^{(1, t+1)}_n=\tilde E^{(1+m, t)}_n$ if the carrier
  has a sufficient number of empty spaces.
  Further, by the discussion in the previous paragraph,
  we find that the sum of
  ``the number of empty boxes that the carrier passes with no balls between the $n$th and ($n+1$)st blocks of balls''
  and ``the number of empty spaces in the carrier just before passing the ($n+1$)st block
  of balls'' is equal to
  \begin{flalign*}
    \max(0, F^{(1, t)}_n-S^{(t)})+\min(F^{(1, t)}_n, S^{(t)})
    &=F^{(1, t)}_n-\min(F^{(1, t)}_n, S^{(t)})+\min(F^{(1, t)}_n, S^{(t)})&\\
    &=F^{(1, t)}_n.
  \end{flalign*}
  Therefore, the value of $F^{(1, t)}_n$ also gives
  the maximum number of boxes which can be consist of $n$th block of empty boxes
  at time $t+1$,
  and the recurrence relation~\eqref{eq:nuhtoda-e} indeed calculates
  the value of $E^{(1, t+1)}_n$.
\end{proof}

\section{Concluding remarks}\label{sec:concluding-remarks}
In this paper, we have derived the nuh-Toda lattice
and constructed its particular solution under the finite lattice boundary condition
by using the theory of biorthogonal polynomials.
Further, we have proven that the nuh-Toda lattice is another time evolution
equation of the BBS with many kinds of balls and finite carrier capacity.

Several problems are left for future works.
In subsequent papers, we are going to discuss the following topics.

Firstly, after imposing $(M, 1)$-reduction condition to biorthogonal polynomials,
we have only discussed the time evolution for $t_2$-direction.
However, there is another time variable $t_1$ and
exists another nonautonomous version of the discrete hungry Toda lattice.
We will be able to ultradiscretize the system and consider
a corresponding BBS-like cellular automaton.
Investigating and analyzing this novel cellular automaton,
its solutions, and relations to the BBS discussed in this paper are interesting problems.

Secondly, it is known that the nonautonomous discrete Toda type systems
give good numerical algorithms~\cite{maeda2016gea}.
As mentioned in \sref{sec:m-1-reduction}, the ndh-Toda lattice~\eqref{eq:ndhtoda-sf}
is same as the eigenvalue algorithm for totally nonnegative Hessenberg matrices
proposed by Fukuda \etal~\cite[Algorithm~1]{fukuda2012eam}, and
we have discussed its asymptotic behaviour for a special case
by analyzing the solution~\eqref{eq:tau-expanded-v}.
Then, we have a natural question: Does the recurrence relation of
the ndh-Toda lattice for $t_1$-direction also gives a good numerical algorithm?
By using techniques similar to that used in this paper,
we will be able to construct a particular solution,
perform asymptotic analysis, and give an answer to this question.
Investigating relations between the algorithms of $t_1$-direction and $t_2$-direction
is also an interesting problem.

Thirdly, we will be able to impose other reduction conditions to
the nd-2D-Toda lattice. For example, as a generalization of
the qd algorithm, which is same as the (autonomous) discrete Toda lattice,
the multiple dqd algorithm is proposed by Yamamoto and Fukaya~\cite{yamamoto2009dqd}.
The multiple dqd algorithm is an eigenvalue algorithm for
matrices decomposed to the product of $M_1$ upper bidiagonal matrices
and $M_2$ lower bidiagonal matrices.
We expect that nonautonomous versions of the multiple dqd algorithm
are derived from the nd-2D-Toda lattice by imposing $(M_1, M_2)$-reduction,
i.e. $\mathcal B^{(k_1+M_1, k_2, t_1, t_2)}=\mathcal B^{(k_1, k_2+M_2, t_1, t_2)}$.
There must be many examples and applications not limited to it.

\end{document}